\documentclass[10pt, onecolumn,draft]{IEEEtran}   

\IEEEoverridecommandlockouts
\usepackage{graphics} 
\usepackage{epsfig}

\usepackage{cuted}
\usepackage{amsmath} 

\usepackage{amssymb,amsthm}  
\usepackage{amscd}
\usepackage{bm}
\usepackage[noadjust]{cite}
\usepackage{float}

\newtheorem{lem}{Lemma}
\newtheorem{pro}{Proposition}
\newtheorem{cor}{Corollary}

\newtheorem{exmp}{Example}
\newtheorem{Remark}{Remark}

\newcommand{\tr}{\operatorname{tr}}

\renewcommand{\cal}{\mathcal}
\newcommand{\diag}{\operatorname{diag}}

\renewcommand{\cal}{\mathcal}
\newcommand{\ad}{{\rm ad}}
\newcommand{\0}{{\bf 0}}

\renewcommand{\S}{\mathcal{S}}

\newcommand{\rt}{r_{\operatorname{tot}}}

\newtheorem{Definition}{Definition}

\newtheorem{Theorem}{Theorem}

\newtheorem{Corollary}{Corollary}

\newtheorem{Assumption}{Assumption}

\newcommand{\R}{\mathbb{R}}
\newcommand{\E}{\mathbb{E}}

\renewcommand{\span}{\operatorname{span}}
\renewcommand{\cal}{\mathcal}
\newcommand{\ol}{\overline}

\renewcommand{\(}{\left (}
\renewcommand{\)}{\right )}
\renewcommand{\[}{\left [}
\renewcommand{\]}{\right ]}

\newcommand{\Sp}{\operatorname{Sp}}

\newcommand{\so}{\frak{so}}

\newcommand{\Ch}{\operatorname{Ch}}

\title{\LARGE \bf
Optimal Capacity Allocation for Sampled Networked Systems}

\author{Xudong Chen, M.-A. Belabbas, Tamer Ba\c sar}

\begin{document}
\maketitle
\thispagestyle{empty}
\pagestyle{empty}

\begin{abstract}
We consider the problem of estimating the states of weakly coupled linear systems from sampled measurements.  We assume that the total capacity available to the sensors to transmit their samples to a network manager  in charge of the estimation is bounded above, and that each sample requires the same amount of communication. Our goal is then to find an optimal allocation of the capacity to the sensors so that the average estimation error is minimized. We show that when the total available channel capacity is large,  this  resource allocation problem can be recast as a strictly convex optimization problem, and hence there exists a unique optimal allocation of the capacity. We further investigate how this optimal allocation varies as the available capacity  increases. In particular, we show that if the coupling among the subsystems is weak, then the sampling rate allocated to each sensor  is nondecreasing in the total sampling rate,  and is strictly increasing if and only if the total sampling rate exceeds a certain threshold. 
\end{abstract}

\section{Introduction}

This paper addresses situations in which a network manager is tasked with estimating the state of an ensemble of weakly inter-connected linear systems. For the estimation to be performed, the systems send sampled measurements to the network manager  over a shared communication channel. Because this communication channel has a finite capacity, we seek to optimize the allocation of channel capacity to each sensor in order to minimize the {\it total estimation} error. In this work, we assume that the channel capacity is directly proportional to the number of samples sent by the subsystems; this is verified when all samples are treated equally.

To proceed, we first describe the model adopted in precise terms. We consider  $N$ weakly-coupled stochastic linear systems  with sampled outputs
\begin{equation}\label{eq:ModelWeaklyCoupled1}
S_i:=\left\{
\begin{array}{l}
dx_i = \left (A_ix_i +\epsilon \sum_{j \neq i} A_{ij}x_j \right)  dt + G_i  dw_i \\
y_i(k\tau_0) =  \bar c_i^{\top} x_i(k\tau_0) + v_i(k\tau_0), 
\end{array}
\right. 
\end{equation}
where $1/\tau_0>0$ is the sampling rate of the sensors and $k$ is a positive integer. We have that $A_i, A_{ij} \in \R^{n \times n}$, $\bar c_i \in \R^{n \times p}$, and $G_i \in \R^{n\times n}$, and that $|\epsilon|$ is small. The assumptions that the subsystems have the same state-dimension $n$ and  the outputs $y_i$ have the same dimension $p$ for all $i$, and the assumption that the coupling parameter $\epsilon$ is the same for all pair $(i,j)$, for $i \neq j$, are made to simplify the notations of the paper, but are not necessary for the results to hold. 
The Brownian motions $w_i$ are pairwise independent and the $\nu_i(k \tau_0)$ are pairwise independent normal random variables. The $w_i$ and $\nu_i$ are also assumed to be independent.

We refer to the system described in~\eqref{eq:ModelWeaklyCoupled1} as subsystem~$S_i$. The samples $y_i(k \tau_0)$, $k \in \mathbb{N}$, are  sent over a common channel to a network manager whose objective is to estimate the states $x_i$ of the subsystems $S_i$, for all $i = 1,\ldots, N$, from these samples. The network manager needs to decide the \emph{schedule} with which it receives the samples in order to minimize the estimation error. Note that since the systems are coupled, the knowledge of $y_i$ can help with the estimation of $x_j$, for $i \neq j$.

We now describe in detail the scheduling problem. The network manager has at his disposal $N$ linear sensors from which he can request samples in order to estimate the states of the subsystems. We only consider \emph{periodic} schedules. That is, we assume that over a fixed time period $\tau>0$, the network manager can request up to $\rt=r_1+r_2+\cdots+r_N$ samples from the sensors, where $r_i$ is the number of samples from the $i$-th sensor, bounded below by a positive number $r_{\min}$. 
We thus have $\tau = \rt \tau_0$ and we can assume that the time-period $\tau$ is then divided evenly into $\rt$ time slots. { In each time slot, the network manager can only  have one sample  sent over the common channel from one of the $N$ sensors. Thus, the problem faced by the network manager is to decide how to assign these~$\rt$ slots to the sensors to send their samples to minimize the estimation error.} 

We note here that the problem has {\it two} natural scales, $\tau$ and $\tau_0$ which are proportional to each other, with ratio $\rt$. We use the following notation to refer to the time slots: The sub-index $l$ refers to the current position within a time period, and the main index $k$ refers to the current period. More specifically, for an arbitrary time signal $s(t)$, we let 
\begin{equation}\label{eq:deftimeslotsnotation}s_{(l)}(k):= s((k\rt+l)\tau_0)\end{equation} where $l$ is only allowed to take values in the set $\{0,\ldots,\rt - 1\}$. With this convention, we can write  the output of the $i$-th sensor as
\begin{equation}\label{eq:sensors2}
y_{i,(l)}(k) = \bar c_i^\top x_{(l)}(k) +v_{i,(l)}(k),
\end{equation} 
where the $\nu_{i,(l)}(k)$'s are pairwise  independent normal variables. 

We call an {\bf allocation strategy} an assignment of the time slots to the sensors over a period $\tau$, and denote by $\cal R$ the set of all possible allocation strategies. We call $\cal{R}$ the {\bf strategy set}.  Our objective  is thus to find the allocation strategy that minimizes the  time-averaged (infinite horizon) estimation error. We refer to this problem as the {\bf optimal allocation problem}.   { A precise formulation  of the problem is presented in Section~\ref{sec:prelim}.}

{
The optimal allocation problem (also known as the optimal scheduling problem if the dynamics for the state~$x$ is in discrete-time) has been investigated for decades, with numerous applications in networked control and estimation, such as localization of mobile robot formations~\cite{mourikis2006optimal}, navigation of underwater vehicles using sonar sensors~\cite{meduna2008low}, target tracking~\cite{he2006sensor}, trajectory planning~\cite{singh2007simulation}, to name just a few. Because of its widespread relevance, there has been continuing efforts in designing efficient algorithms for finding the optimal (but often an suboptimal) solution to the allocation/scheduling problem. Amongst the related works, we first note the seminal work~\cite{meier1967optimal} by Meier, Peschon, and Dressler: The authors there consider a {\em discrete-time} linear control system with multiple sensors. But only one  sensor can be used at each time step.   The objective is thus to determine the schedule of the sensors  to minimize the total estimation error for a finite horizon. The optimization problem is then solved via dynamic programming. However, such a method is often computationally intractable especially when the number of sensors is large and  schedule horizon is long (here, $N$ and~$\rt$ are large). Following~\cite{meier1967optimal}, there have been various methods established to reduce the computational complexity. Among the deterministic methods, greedy algorithms have been used several times  to find suboptimal solutions (see, for example~\cite{oshman1994optimal,kagami2006sensor,chhetri2007use}). Other algorithms, such as pruning of the tree-search, have also been proposed (note that the optimal scheduling problem is a special type of tree-search problem). For example, the tree-pruning algorithms established in~\cite{vitus2012efficient} trade-off the quality of the solution and the complexity of the problem through a tuning parameter. We further refer to~\cite{alriksson2005sub} for an suboptimal algorithm using relaxed dynamic programming.   
Besides deterministic algorithms, there are also stochastic methods dealing with computational complexity of the optimal scheduling problem. For example, the authors in~\cite{gupta2006stochastic} select a sensor randomly at each time step according to a certain probability distribution. An upper-bound for the expected value of the stead-state estimation error is established. The probability distribution is then chosen so as to minimize the upper-bound. For other stochastic methods, we refer to~\cite{he2006sensor} for a Monte Carlo method, and to~\cite{singh2007simulation} for a simulation-based approach.  We further point out that the optimal scheduling problem is also investigated for nonlinear processes. For example, the authors established in~\cite{baras1989optimal} the existence of an optimal solution for nonlinear diffusion processes.     

Amongst other related works, there have also been studies on a similar problem called {\em optimal sensor selection}, for which the objective is to select a relatively small subset of sensors to be put to use at each time step so as to minimize the estimation error. The optimal selection problem also faces the challenge of high computational complexity; indeed, it has been proved in~\cite{zhang2015sensor} that the problem is NP-hard, which holds even if the system is stable. Various algorithms have also been proposed to deal with the computational complexity. We refer to~\cite{joshi2009sensor} for an approach using convex relaxation, and to~\cite{isler2005sensor} for an approximation algorithm which finds a suboptimal solution in a polynomial time and guarantees that the resulting estimation error is within factor 2 of the least possible error. 

For more works related to the optimal scheduling problem, we note that there are settings where there exist energy constraints and/or running costs for  transmitting samples, and the objective is thus to decide whether to transmit or not~\cite{gao2015optimal,martins2007fundamental,gupta2009optimal,imer2005optimal,imer2010optimal,bommannavar2008optimal,imer2005ACM}. There are also settings where the channels over which the samples are sent are lossy~\cite{sinopoli2004kalman,sinopoli2008optimal,imer2006optimal,moon2015minimax,bommannavar2008optimal}. While our set-up is related in spirit, the approach required to solve these problems is  different from the one we need here.

As mentioned above, most extant work in this area has dealt with the computational complexity by appealing to heuristic algorithms, and look for suboptimal solutions. In this work, we investigate the optimal allocation problem from a different perspective: First, we recall that the dynamics of networked system is in continuous-time, with $N$ sensors sampling the state at a rate of $1/\tau_0$. A total number of $\rt$ samples can be obtained in a scheduling period~$\tau$. We investigate in this work how the estimation error depends on the sampling rate, and moreover, how such a dependence affects the solution to the optimal capacity allocation problem. One of our main contributions is then to show that the optimal allocation problem can be solved {\em exactly} for $\tau$ asymptotically small (or equivalently, the sampling rate asymptotically high).  We note here that the question about the dependence of the estimation error on the sampling rates, and the optimal allocation problem, have also been addressed recently in the computer science  and cyber-physical systems communities~\cite{lu2016real,seto2001trade,saifullah2014near}. This line of work, however,  relies on a heuristic claim that the performance measure (here, the total estimation error) decays exponentially  in the sampling rate. We show in the next section that this is not in fact true, and derive the exact asymptotic behavior.    
}

We now briefly outline the approach taken in this paper. First, we show that the estimation error afforded by a given allocation strategy $R$---in the appropriate asymptotic limit---is \emph{independent} of the order in which the samples are requested, but depends only on the total numbers of samples requested from each sensor in a time period $\tau$ by the strategy $R$. This simplifies the problem greatly and allows us to show that it is in fact \emph{equivalent} to a continuous-time estimation problem where the limited resource is not the channel capacity, but the {\it quality} of the sensors, specifically, the signal-to-noise ratio of the measurements they provide.  Said otherwise, we show that under a few natural assumptions, we can replace the generally difficult problem of optimally assigning time slots to sensors with the easier problem of optimally choosing signal-to-noise ratios of measurements. 

We  state these problems and show their equivalence in Section~\ref{sec:prelim}.
Next, in Section~\ref{sec:optimalsol}, we study the dependence of the error covariance on the sampling rate for a {\it single} subsystem. In particular, we show that it is {\it strictly convex} and {\it monotone decreasing} in the sampling rate. We then show that if several subsystems are weakly coupled, the resulting optimal capacity allocation problem is also strictly convex and hence admits a unique {\em optimal allocation}.  Finally, in Section~\ref{sec:monotone}, we investigate how this optimal allocation depends on the total available capacity.

\noindent{\bf Definitions and notations.} We describe the notation used throughout the paper.
We denote by $\R^N_+$ the nonnegative orthant in $\R^N$. Given $\sigma >0$, we define the {\bf simplex of height $\sigma$} in $\R^N$ as 
\begin{equation}\label{eq:defsimplex}
\Sp[\sigma] := \left\{ v\in \R^{N} \mid v_i  \ge 0 \mbox{ and } \sum^N_{i=1}v_i = \sigma  \right\}.
\end{equation}
We use the acronym ARE to refer to the algebraic Riccati equation~\cite{brockett2015finite}. For a matrix $P$, we let $\tr(P)$ be the trace of $P$. For a symmetric matrix $P$, we write $P \ge 0$ (resp. $P \le 0$) if $P$ is positive (resp. negative) semi-definite, and $P > 0$ (resp. $P < 0$) if $P$ is positive (resp. negative) definite. 
An $N \times N$ matrix $P = (p_{ij})$ is said to be  {\bf diagonally dominant} if 
$$
|p_{ii}| \ge \sum_{j\neq i} |p_{ij}|, \hspace{10pt} \forall\, i = 1,\ldots. N.
$$
Given a square matrix $P$,  a {\bf principal} submatrix of $P$ is a matrix derived by removing certain rows and columns of $P$, with the condition that the two sets of indices---the indices of the rows that are removed and the indices of the columns that are removed---are the same.

\section{Problem Formulation}\label{sec:prelim}
\subsection{The optimal allocation problem}

We present the estimation problem in the case of a single system with $N$ sensors. 
To this end, we consider a linear stochastic system with $N$ sensors:
\begin{equation}\label{eq:Model1Cont}
\left\{
\begin{array}{l} 
dx = A  x dt + G  dw\\
y_{i,(l)}(k) = \bar c_i^\top x_{(l)}(k) +v_{i,(l)}(k), \hspace{10pt} i = 1,\ldots, N.
\end{array}\right.
\end{equation} 
Recall that the optimal allocation problem consists of assigning the $\rt$ time slots in a period $\tau = \rt \tau_0$ to the sensors in order to minimize the estimation error of $x(t)$. 
For a later purpose, we note here that the optimization problem is comprised of the following two inter-related problems:
\begin{enumerate}
\item[P1).] Given $\rt >0$ fixed, determine how many time slots $r_i \ge 0$  are assigned to sensor $i$, subject to the constraint that $\rt=r_1+r_2+\cdots r_N$ and $r_i\ge r_{\min}$ for all $i = 1,\ldots, N$? 
\item[P2).] Given $r_1,\ldots,r_N \ge 0$ fixed, determine how to assign the $r_i$ slots to the $i$-th sensor, for all $i = 1,\ldots, N$, out of the total $\rt$ slots?
\end{enumerate}

Note that we do {\it not} consider here feedback strategies of allocations, in which, for example, the network manager decides which sensor should send its sample for the upcoming time slot based upon all the past observations. Of course, such a feedback strategy would evidently yield a better performance, but their real-time implementation is far more difficult. We instead focus on the optimal {\it design} problem, for which the network manager makes a static assignment that is used for every period $\tau$. 
We can assume that all the sensors sample their outputs at the same frequency $1/\tau_0$ and that the network manager requests the samples as needed. 


To proceed, let $\hat x_{(l)}(k)$ be the optimal {\it mean squared error} (MSE) estimate of the state $x_{(l)}(k)$ of the system  by the network manager. It is well known that the MSE estimate is the conditional expectation of $x_{(l)}(k)$ given all the past observations.  It is also well known how to update the MSE estimate recursively (see, for example~\cite{optimalfiltering_anderson79}). We thus only  sketch the recursive derivation with an eye towards obtaining asymptotics: First, let $e_{(l)}(k)$ be the error in estimation of $x_{(l)}(k)$: $$e_{(l)}(k) := x_{(l)}(k)-\hat x_{(l)}(k).$$  We denote the corresponding error covariance as follows: $$ \Sigma_{(l)}(k):=\E\left[e_{(l)}(k)e_{(l)}^\top(k) \right],$$ where the expectation is conditioned on the past observations. We note here that the trace of $\Sigma_{(l)}(k)$, denoted by $\tr(\Sigma_{(l)}(k))$, is then the estimation error. Now, by first taking the expectation on both sides of the evolution equation~\eqref{eq:Model1Cont} and then integrating over one time slot $\tau_0$, we obtain the MSE estimate of $x_{(l)}(k)$ without using the new sample $y_{i,(l)}(k)$, for some $i\in \{1,\ldots, N\}$, which is solved by  \begin{equation*}\label{eq:xmupdate}\hat x_{(l)}^-(k) := e^{A \tau_{0}} \hat x_{(l-1)}(k).\end{equation*}
Correspondingly, the covariance $\Sigma_{(l)}^-(k)$ of the estimation error $e_{(l)}^-(k):=\hat x_{(l)}^-(k) - x_{(l)}(k)$ is obtained by integrating the following Lyapunov differential equation:
$$\dot \Sigma = A \Sigma + \Sigma A^\top + GG^\top$$ over a time slot $\tau_0$, with $\Sigma_{(l-1)}(k)$ the initial condition. The solution can be obtained explicitly as follows:  
\begin{equation}\label{eq:covupdate}
\Sigma_{(l)}^-(k)  =   e^{A\tau_0}\Sigma_{(l-1)}(k)e^{A^\top\tau_0} + 
\displaystyle\int_0^{\tau_0} e^{As} GG^\top e^{A^\top s}ds.
\end{equation}
Upon receiving the new sample $y_{i,(l)}(k)$, we update the mean and covariance as follows:
\begin{equation*}\label{eq:xupdate}
\hat x_{(l)}(k)  =   \hat x_{(l)}^-(k) + \Sigma_{(l)}^-(k)\bar c_i  \Big [\bar c_i^\top\Sigma_{(l)}^-(k)\bar c_i+I \Big ]^{-1} \\ \(y_{i,(l)}(k)-\bar c_i^\top \hat x_{(l)}^-(k)\),
\end{equation*}
and 
\begin{equation}\label{eq:covplusupdate}\Sigma_{(l)}(k) = \left(\Sigma_{(l)}^-(k)^{-1} + \bar c_i \bar c_i^\top\right)^{-1}.\end{equation}
For convenience, we define the map $$\phi_i: \Sigma_{(l-1)}(k) \mapsto \Sigma_{(l)}(k)$$ which sends an error covariance matrix  to its update over a single slot. The sub-index~$i$ indicates that sensor~$i$ is used in the corresponding time slot.

{
We recall that $\cal R$ is the set of allocation strategies of assigning the slots to the sensors in a scheduling period~$\tau$. Let $R\in \cal R$ be an allocation strategy. Note that $R$ can be represented by a $\rt$-dimensional vector, whose entries take values in the set $\{1,\ldots, N\}$. More specifically, if the $l$-th entry of $R$, denoted by $R_l$, is~$i$, then sensor~$i$ is used at the $l$-th slot over the period $\tau$. 
We now associates each allocation strategy $R$ a map $\Phi_R$, defined as the composition of $\phi_{R_l}$ for $l = 1,\ldots, \rt $: 
$$
\Phi_R := \phi_{R_{\rt}}\cdots \phi_{R_1}.
$$
Note that the map $\Phi_R$ depends (implicitly) on the scheduling period $\tau$, and the matrices~$A$, $G$ and $\bar c_i$, for $i = 1,\ldots, N$, in~\eqref{eq:Model1Cont}.    
Further, for ease of notation,  we omit the sub-index of $\Sigma_{(l)}(k)$ if $l = 1$. 
Then, with the map $\Phi_R$ defined above,  we have that   
$
\Sigma(k) = \Phi_{R}(\Sigma(k-1))  
$. 
Thus, given an initial condition $\Sigma(0)$ of the error covariance, we can obtain  $\Sigma(k) = \Phi^k_R(\Sigma(0))$ for all $k \ge 0$.

We now establish a sufficient condition for the convergence of the sequence of error covariance matrices. First, let ${\rm Eig}(A)$ be the set of eigenvalues of the matrix $A$ in~\eqref{eq:Model1Cont}. 
We then let ${\cal T}$ be a subset of positive numbers $\tau$ defined as follows: if $\tau\in \cal{T}$, then for any pair of distinct eigenvalues $(\lambda_i,\lambda_j)$ of $A$, we have that $e^{\lambda_i\tau} \neq e^{\lambda_j \tau}$, i.e.,  
$$
{\cal T} := \{\tau \in \R_+ \mid e^{\lambda_i\tau} \neq e^{\lambda_j \tau} \mbox{ if } \lambda_i \neq \lambda_j \in {\rm Eig}(A)\} 
$$
It should be clear that $\cal T$ is an open dense subset of $\R_+$. Note, in particular, that if $\tau$ is sufficiently small such that 
\begin{equation}\label{eq:smalltau}
|{\rm Im}(\lambda_i) - {\rm Im}(\lambda_j) |\tau < 2\pi, \hspace{10pt}  \forall \lambda_i, \lambda_j\in {\rm Eig}(A),
\end{equation}
where ${\rm Im}(\cdot)$ denotes the imaginary part of a complex number,   
then $\tau \in {\cal T}$.      
With the definitions and notations above, we now have the following fact:

\begin{pro}\label{pro:existenceoflimit}
Let  $\ol c := [\bar c_1, \ldots, \bar c_N]$, with $\ol c_i$ in~\eqref{eq:Model1Cont}. Suppose that $(A, \ol c)$ is an observable pair; then, for any allocation strategy $R\in \cal R$ and any scheduling periodic $\tau\in \cal T$, the sequence of error covariance matrices converges to a steady state: 
$$ \Sigma(\infty; \tau, R):=\lim_{k\to\infty} \Phi^k_R(\Sigma(0)),$$        
which depends only on $\tau$ and $R$, but not on the initial condition $\Sigma(0)$.  
\end{pro}

We refer to Appendix A for a proof of the proposition. 
In the sequel, we assume that $\tau \in \cal T$. Following Proposition~\ref{pro:existenceoflimit}, we formalize the optimal allocation problem as the problem of minimizing the steady-state estimation error: 
\begin{equation}\label{eq:defeta}
\ol \eta(\tau, R) := \tr(\Sigma(\infty; \tau, R)), 
\end{equation} 
over all possible allocation strategies~$R \in \cal{R}$.
}

\subsection{The small $\tau$ asymptotic}\label{ssec:asympt}

The optimal allocation problem described in the previous section is a combinatorial optimization problem, and is in general hard to solve, especially when both $N$ and $\rt$ are large.  We show in this subsection that in the $\tau \to 0$ asymptotic, major simplifications occur that ultimately lead us to a (strict) convex optimization problem which is thus tractable.  Specifically, we will show that when $\tau \to 0$, the exact order in which the measurement signals $y_i$ are sampled is not important, but only the number of times they are sampled in a period $\tau$ matters. In other words, the network manager only needs to solve Problem P1 in Subsection II-A, i.e. determine how many slots are allocated to each sensor, and {\em not} Problem P2, i.e. determine which slots are allocated to each sensor.

To proceed, we first note the following fact: When the sampling period of the measurement signal decreases, as a compensation, we need to let the norm of the observation vector $\bar c_i$ decrease (or, equivalently, let the variance of the measurement noise $v(k\tau)$ increase) at a rate  proportional to the {\it square root} of the sampling period. We thus write \begin{equation}\label{eq:barc}\bar c_i =   c_i \sqrt{\tau},\end{equation} for some vector $ c_i$ with fixed norm. This  scaling applies to all observation vectors $\bar c_i$, for $i = 1,\ldots, N$. 

{
We now fix an allocation strategy $R\in \cal{R}$, and focus on the update equation $\Sigma(k+1) = \Phi_R(\Sigma(k))$. Recall that the map $\Phi_R$ is defined as the composition of $\phi_{R_l}$ for $l = 1,\ldots, \rt$. Thus, to obtain from $\Sigma(k)$ to $\Sigma(k+1)$, we need to recursively  apply $\phi_{R_l}$ to obtain $\Sigma_{(l+1)}(k)$ from $\Sigma_{(l)}(k)$.  To this end, we assume that $R_1 = i$, i.e., the sensor~$i$ is used at the first slot. Then, by appealing to~\eqref{eq:covupdate} and~\eqref{eq:covplusupdate} and the scaling $\bar c_i \bar c_i^\top = \tau c_ic_i^\top$ defined in~\eqref{eq:barc}, we obtain that
\begin{equation}\label{eq:smallstepupdate}
\Sigma_{(1)}(k) = \Sigma_{(0)}(k) +   \Sigma_{(0)}(k) c_ic_i^\top \Sigma_{(0)}(k) \tau + \\\left(A\Sigma_{(0)}(k)+\Sigma_{(0)}(k) A^\top  +G G^\top\right)\tau_0  + {\rm o}(\tau),
\end{equation}
where ${\rm o}(\tau)$ denotes the higher order terms in~$\tau$. Similarly, if the sensor~$j$ is used in the next slot, i.e.,  $R_{2} =j $, then 
\begin{equation}\label{eq:smallstepupdate1}
\Sigma_{(2)}(k) = \Sigma_{(1)}(k) +   \Sigma_{(1)}(k) c_jc_j^\top \Sigma_{(1)}(k) \tau + \\\left(A\Sigma_{(1)}(k)+\Sigma_{(1)}(k) A^\top  +G G^\top\right)\tau_0  + {\rm o}(\tau).
\end{equation} 
We now plug~\eqref{eq:smallstepupdate} into~\eqref{eq:smallstepupdate1}. Then, up to the first order in~$\tau$, we obtain
\begin{equation}\label{eq:smallstepupdate2}
\Sigma_{(2)}(k) = \Sigma_{(0)}(k) +  \Sigma_{(0)}(k) (c_ic_i^\top + c_jc_j^\top) \Sigma_{(0)}(k) \tau + \\ 2\left(A\Sigma_{(0)}(k)+\Sigma_{(0)}(k) A^\top  + G G^\top\right)\tau_0  + {\rm o}(\tau).
\end{equation} 
By repeating this process for $\rt$ times, we then obtain $\Sigma(k+1)$ from $\Sigma(k)$ as follows:   
\begin{pro}\label{pro:pronice}
Let $R$ be an allocation strategy,  and $r_i$ be the number of slots times assigned to sensor~$i$ over a scheduling period~$\tau$. Then, 
\begin{equation}\label{eq:updateequation}
\Sigma(k+1) =  \Sigma(k) +  \tau \Big [ A\Sigma(k)+\Sigma(k) A^\top + \\ \sum_{i=1}^N r_i\Sigma(k) c_ic_i^\top \Sigma(k) +G G^\top \Big ] + {\rm o}(\tau).
\end{equation}
\end{pro}

Now, let $\tau$ be sufficiently small such that~\eqref{eq:smalltau} holds. Then, from Proposition~\ref{pro:existenceoflimit}, the steady state $\Sigma(\infty; \tau, R)$ exists, and satisfies the following expression: 
$$
\Sigma(\infty; \tau, R) = \phi_R(\Sigma(\infty; \tau, R)).
$$

It thus follows from Proposition~\ref{pro:pronice} that if we express $\Sigma(\infty; \tau, R)$ as $$\Sigma(\infty; \tau, R) = \Sigma^{(0)}(R) + \tau \Sigma^{(1)}(R) + \cdots, $$ then its zeroth order term $\Sigma^{(0)}(R)$ must satisfy the following algebraic Ricatti equation (ARE):
\begin{equation}\label{eq:algriccattiforzeroth}
A\Sigma^{(0)}(R) + \Sigma^{(0)}(R) A^\top  \\ - \sum_{i=1}^N r_i \Sigma^{(0)}(R) c_ic_i^\top\Sigma^{(0)}(R) +GG^\top = 0.  
  \end{equation}
Note that if $(A, \ol c)$ is an observable pair, then so is $(A, c)$, with $c := [c_1,\ldots, c_N] = \ol c \sqrt{\tau} $. Thus, $\Sigma^{(0)}(R)$ is the unique positive semi-definite solution to~\eqref{eq:algriccattiforzeroth}.    
We further note that $\Sigma^{(0)}(R)$ depends only on the $r_i$'s, i.e., the numbers of slots assigned to the sensors, but {\em not} on the order in which we sample these sensors.   

We conclude this section with the following fact which will be essential to the analysis of the optimal allocation problem: For an allocation strategy $R$, we consider the following system with continuous observation signals (compared to~\eqref{eq:Model1Cont}): 
\begin{equation}\label{eq:ModelWeaklyCoupledpre}
 \left\{
\begin{array}{l}
dx = Ax   dt + G  dw \\
dz_i =  \sqrt{r_i} c_i^{\top} x dt + d\nu_i, \hspace{10pt}  i = 1,\ldots, N.
\end{array}
\right. 
\end{equation} 
The MSE estimate of its state given the past observations is given by the Kalman-Bucy filter, and if we denote by $\Sigma(t)$ the covariance of the estimation error at time $t$, then it is well known that $\Sigma$ obeys the following continuous-time Riccati differential equation: 
\begin{equation}\label{eq:eqXX1}
\dot \Sigma = A\Sigma +\Sigma A^\top -\sum_{i=1}^N r_i \Sigma c_i c_i^\top \Sigma +GG^\top.\end{equation}
Furthermore, since $(A,c)$ is an observable pair, the solution of~\eqref{eq:eqXX1}, with any initial condition $\Sigma(0)\ge 0$,  converges to a unique positive semi-definite matrix $\Sigma(\infty)$, which satisfies the same equation~\eqref{eq:algriccattiforzeroth} as $\Sigma^{(0)}(R)$ does. We thus obtain the following result: 

\begin{Corollary}\label{cor:cornice}
In the limit $\tau \to 0$, the steady-state estimation error, as defined in~\eqref{eq:defeta},  of the sampled system~\eqref{eq:Model1Cont} for a fixed allocation strategy~$R$ coincides with the estimation error of system~\eqref{eq:ModelWeaklyCoupledpre}: $$\lim_{\tau \to 0} \ol \eta(\tau;R) = \tr(\Sigma(\infty)).$$\, 
\end{Corollary}

\section{The Optimal Solution for Resource Allocation} \label{sec:optimalsol}
\subsection{On Algebraic Riccati Equations}
In this subsection, we establish the properties of Algebraic Riccati Equation (ARE) that are needed to prove the results of this paper. For $r>0$,
we introduce the following ARE:   
\begin{equation}\label{eq:RiccatiEquation}
A\Sigma + \Sigma A^\top - r\Sigma cc^\top \Sigma + GG^\top = 0, 
\end{equation}
It is well known that if $(A,c)$ is observable, then~\eqref{eq:RiccatiEquation} admits a unique positive semi-definite matrix $\Sigma$ as its solution. Furthermore, if $GG^\top > 0$, then $\Sigma > 0$. 
It should be clear that the positive semi-definite solution $\Sigma$ to ARE~\eqref{eq:RiccatiEquation} is nothing but the steady state of the Riccati differential equation~\eqref{eq:eqXX1}. Thus, the trace of $\Sigma$, denoted by $\tr(\Sigma)$, is the MSE in steady state for the Kalman-Bucy filter of system~\eqref{eq:ModelWeaklyCoupledpre}. We will sometimes emphasize the dependence of $\Sigma$ on $r$ by  writing $\Sigma(r)$ explicitly. We further denote the first and second derivatives (taken entry-wise) of $\Sigma$ with respect to $r$ as 
$$\Sigma'(r) := \frac{d\, \Sigma(r)}{ d r} \hspace{10pt} \mbox{ and } \hspace{10pt} 
\Sigma''(r) := \frac{d^2\, \Sigma(r)}{dr^2}. $$ 
We investigate below the dependence of $\Sigma(r)$ on $r$. We start with the following definition: 

\begin{Definition}[Regular triplet]\label{def:regulartriplet}
Let $(A, c)$ be an observable pair, and $P$ be a positive definite matrix. We say that $\(A, c, P\)$ is a {\bf regular triplet} if the following  condition is satisfied: 
there is a number $r >0$ such that the pair $(A,\Sigma c)$ is controllable, where  $\Sigma$ is the unique positive definite solution of the following ARE:
\begin{equation}\label{eq:AREinDef1}
A \Sigma + \Sigma A^\top - r \Sigma cc^\top \Sigma + P = 0.
\end{equation}\,     
\end{Definition}

We show below that if a triplet $(A,c,P)$ is regular for some $r>0$, then it is regular for \emph{all} $r>0$. Precisely, we have the following fact:

\begin{pro}\label{pro:oneforall}
Suppose that $(A, c, P)$ is a regular triplet, with $(A, c)$ observable and $P > 0$; then, for any positive number $r$, the pair $(A, \Sigma(r)c )$, with $\Sigma(r)$ the unique positive definition solution to~\eqref{eq:AREinDef1}, is controllable. 
\end{pro}

We refer to Appendix~B for a proof of Proposition~\ref{pro:oneforall}. 
We note here that not all triplets $(A, c, P)$, with $(A,c)$ observable and $P> 0$, are regular. An illustration of a counter example is given below: 
\begin{exmp}
Let a triplet $(A, c, P)$ be given by 
$$
\begin{array}{lll}
A = 
\begin{pmatrix}
-1 & 0\\
0 & -2
\end{pmatrix}, \, & 
c = 
\begin{pmatrix}
1 \\
1
\end{pmatrix}, \, &
P = 
\begin{pmatrix}
5 & -3\\
-3 & 4
\end{pmatrix}. 
\end{array}
$$
One sees that $(A, c)$ is observable and $P > 0$. Let $r = 1$; then, the positive definite solution to~\eqref{eq:AREinDef1} is given by
$$
\Sigma = 
\begin{pmatrix}
2 & -1 \\
-1 & 1
\end{pmatrix}.
$$   
But the pair
$$
(A, \Sigma c) = \( 
\begin{pmatrix}
-1 & 0\\
0 & -2
\end{pmatrix},    \, 
\begin{pmatrix}
1 \\
0
\end{pmatrix}
\)
$$
 is not controllable. 
\end{exmp}

However, we show in Appendix~B that a triplet $(A, c, P)$, with $(A,c)$ observable and $P > 0$,  is {\it generically} regular.  
With the preliminaries above, we state the following fact:

\begin{pro}\label{lem:SigmaDerivative}
Let  $(A,c)$ be an observable pair and $GG^\top>0$. Let $\Sigma(r)$, for $r > 0$, be the positive definite solution to 
\begin{equation}\label{eq:LyapunovEqSingle1}
A \Sigma(r) + \Sigma(r) A^\top - r \Sigma(r) cc^\top \Sigma(r) + GG^\top = 0.
\end{equation}  
Then, $$\Sigma'(r) \le 0 \hspace{5pt} \mbox{ and }  \hspace{5pt} \Sigma''(r) \ge 0.$$
Furthermore, if the triplet $\(A, c, GG^\top\)$ is regular, then the inequalities above are strict. 
\end{pro}

We refer to Appendix~C for a proof of Proposition~\ref{lem:SigmaDerivative}.  

\begin{Remark}
Note that a similar result has been derived in~\cite{wredenhagen1993curvature}. Specifically, the authors there consider the following ARE: 
$$
A \Sigma + \Sigma A^\top - \Sigma S \Sigma + P = 0,
$$ 
and they have shown that the solution $\Sigma$ is convex in $S$ and concave in $P$. In Proposition~\ref{lem:SigmaDerivative}, we  provide in addition a sufficient but generic condition for $\Sigma(r)$ to be a strictly convex function in~$r$. 
\end{Remark}
With a slight abuse of notation, we let
 \begin{equation}\label{eq:optimaleta}
\eta(r) := \tr(\Sigma(r))
\end{equation}
and set $$
\eta'(r) := \frac{d \eta(r)}{d r} \hspace{10pt} \mbox{and} \hspace{10pt} \eta''(r) := \frac{d^2 \eta(r)}{d r^2}.$$ 
As an immediate consequence of Proposition~\ref{lem:SigmaDerivative}, we have

\begin{cor}\label{cor:OptimalEtaDer}
Let $\(A, c, GG^\top\)$ be a regular triplet, and $\eta(r)$ be as in~\eqref{eq:optimaleta}.    Then, 
$$\eta'(r) < 0 \hspace{5pt} \mbox{ and } \hspace{5pt}  \eta''(r) > 0.$$
\,
\end{cor}
\subsection{On weakly coupled networks}

In this subsection, we return to a network of $N$ weakly-coupled $n$-dimensional stochastic linear systems, denoted by $S_1,\ldots, S_N$:
\begin{equation}\label{eq:ModelWeaklyCoupled2}
S_i= \left\{
\begin{array}{l}
dx_i = \left (A_ix_i +\epsilon \sum_{j} A_{ij}x_j \right) \, dt + G_i \, dw_i, \\
y_{i,(l)}(k) =  \bar c_i^{\top} x_{i,(l)}(k) + v_{i,(l)}(k).
\end{array}
\right. 
\end{equation}
Using Corollary~\ref{cor:cornice}, we know that for a fixed allocation strategy $R \in {\cal R}$,  the MSE estimation problem for the sampled system~\eqref{eq:ModelWeaklyCoupled2} is equivalent, in the asymptotic case $\tau \to 0$, to the MSE  estimation problem for the continuous system below:
\begin{equation}\label{eq:ModelWeaklyCoupled}
S_i= \left\{
\begin{array}{l}
dx_i = \left (A_ix_i +\epsilon \sum_{j} A_{ij}x_j \right) \, dt + G_i \, dw_i \\
dz_i =  \sqrt{r_i} c_i^{\top} x_i dt + d\nu_i,
\end{array}
\right. 
\end{equation}
where we recall $c_i = \bar c_i / \sqrt{\tau} $ (see~\eqref{eq:barc}), and  $r_i$ is the number of time slots assigned to the sensor~$i$ over a time period~$\tau$ by the allocation strategy $R$.  
For the remainder of the paper, we assume that the sub-systems in~\eqref{eq:ModelWeaklyCoupled} satisfy the following mild assumption:

\begin{Assumption}\label{asmp:fornetworksystem}
Each sub-system~$S_i$, for $i = 1,\ldots, N$,  satisfies the following condition: The pair $(A_i, c_i)$ is observable, $G_i$ is nonsingular, and the triplet $\(A_i, c_i, G_iG^\top_i \)$ is regular. 
\end{Assumption}

To proceed, we first introduce some notations that will make the derivations of the subsequent results easier. 
Let $A$ and $G$ be two $nN\times nN$ matrices defined as follows: 
\begin{equation}\label{eq:defANetSys}
A:= 
\begin{pmatrix}
A_{1} & \epsilon A_{12} & \ldots & \epsilon A_{1N}\\
\epsilon A_{21} & A_{2} & \ldots & \epsilon A_{2N}\\
\vdots & \vdots & \ddots & \vdots \\
\epsilon A_{N1} & \epsilon A_{N2} & \ldots & A_{N} 
\end{pmatrix}, 
\end{equation} 
and 
$$G= \diag(G_1,\ldots, G_N).$$  
Let $x := (x_1,\ldots, x_N) \in \R^{nN}$  be defined by concatenating $x_i$, for $i = 1,\ldots, N$. Similarly, we define $dw:= (dw_1,\ldots, dw_N)$. With the notations above, we can re-write the network dynamics~\eqref{eq:ModelWeaklyCoupled} as follows:
\begin{equation}\label{eq:XxX4}
	\left\{
\begin{array}{l}
dx = Ax dt  + G dw, \\
dz_i =  \sqrt{r_i} c_i^{\top} x_i dt + d\nu_i, \hspace{10pt} \forall\, i = 1,\ldots, N.
\end{array}
\right. 
\end{equation}
Recall that $r_{\min} > 0$ is the lower bound such that $r_i \ge r_{\min} $ for all $i = 1,\ldots, N$. So, for ease of notation, let $\gamma$ be a vector in $\R^N$ defined as follows:    
\begin{equation}\label{eq:defgamma}
\gamma := (\gamma_1,\ldots, \gamma_n), \hspace{10pt} \mbox{ with } \hspace{5pt} \gamma_i:= {r}_i - r_{\min}.
\end{equation}
It should be clear that $\gamma$ lies in $\R^N_+$, the nonnegative orthant of $\R^N$. 
We gather all the observation vectors $\sqrt{r_i} c_i$ in a matrix $c_{\gamma}$ as follows:
$$
c_{\gamma}:= \diag(\sqrt{\gamma_1 + r_{\min}}c_1,\ldots, \sqrt{\gamma_N + r_{\min}} c_N).
$$
Now, let $dz := (dz_1,\ldots, dz_N)$ be the concatenation of all the observations, and similarly, let $d\nu:= (d\nu_1,\ldots, d\nu_N)$. 
We then further simplify~\eqref{eq:XxX4} as follows:    
\begin{equation}\label{eq:MMODEL}
\left\{
\begin{array}{l}
dx = Ax dt  + G dw, \\
dz = c_\gamma^\top x dt + d\nu. 
\end{array}
\right.
\end{equation}
The network dynamics is thus parametrized by the vector $\gamma\in \R_+^N$. To proceed, we first note the following fact about the observability of the pair $(A, c_{\gamma})$: 

\begin{lem}
Let $c := \diag(c_1,\ldots, c_N)$. If $(A, c)$ is an observable pair, then so is $(A,c_\gamma)$ for all $\gamma \in \R^N_+$. Furthermore, if Assumption 1 holds, then there exists $\epsilon_0$ such that for all $\epsilon$ with $|\epsilon| \leq |\epsilon_0|$, the pair $(A,c)$, and hence the pair $(A,c_\gamma)$, is observable
\end{lem}

\begin{proof}
The proof of the first statement directly follows from the definitions of $c$ and $c_\gamma$, as well as the expressions of the observability matrices for the pairs $(A,c)$ and $(A,c_\gamma)$. For the second part, note that from Assumption 1, $(A_i,c_i)$ is observable for all $i = 1,\ldots, N$, and hence $(A,c)$ is observable for $\epsilon=0$. Since being observable is an open condition, the statement follows.
\end{proof}
For the remainder of the paper, we shall assume  that $\epsilon\in [0,\epsilon_0]$, which then implies that $(A, c_\gamma)$ is observable for all $\gamma \in \R^N_+$.  
Now, for a given $\gamma\in \R^N_+$, let $\Sigma$ be the (unique) positive definite solution of the ARE:
\begin{equation}\label{eq:RiccatiEquationNetSys}
A\Sigma + \Sigma A^\top - \Sigma c_\gamma c^\top_\gamma \Sigma  + GG^\top = 0.
\end{equation}
As before, we write $\Sigma(\gamma)$ to emphasize the dependence of $\Sigma$ on $\gamma$. Since from our assumptions, $(A, c_\gamma)$ is observable for all $\gamma\in \R^N_+$,  $\Sigma(\gamma)$ is well defined for all $\gamma\in \R^N_+$.  In fact, we can establish the following result:

\begin{lem}\label{lem:domainofeta} 
Let $Q$ be a  subset of $\R^N$ defined as 
$$
Q := \{ \gamma = (\gamma_1,\ldots, \gamma_N) \mid \gamma_i  \ge - r_{\min}\}.  
$$ 
Then, the unique positive definite solution $\Sigma(\gamma)$ to~\eqref{eq:RiccatiEquationNetSys} is well defined over $Q$ and is analytic in $\gamma$. 
\end{lem}

\begin{proof}  
Since the pair $(A, c)$, for $c = \diag(c_1,\ldots,c_N)$, is observable, we obtain using the definition of $c_\gamma$ that  so is the pair $(A,c_\gamma)$ if $\gamma_i + r_{\min} > 0$ for all  $i  = 1,\ldots, N$. We thus have that  $\Sigma(\gamma)$ is well defined for all $\gamma\in Q$.  The fact that $\Sigma(\gamma)$ is analytic in $\gamma$ is a consequence of results of~\cite{delchamps_analytic1984}.   
\end{proof}

In the remainder of the section, we solve the optimal allocation problem formalized in Section~\ref{sec:prelim} in the asymptotic case $\tau\to 0$. Specifically,  our  goal is to minimize the steady state mean squared estimation error: 
\begin{equation}
\label{eq:cost1}
\eta(\gamma)= \tr(\Sigma(\gamma)),
\end{equation}
over $\gamma  \in \R^N_+$, subject to the constraint that    \begin{equation}\label{eq:constgamma}\sum^N_{i = 1}\gamma_i\le \sigma:= \rt - Nr_{\min}.\end{equation} 
To proceed, first recall that from Corollary~\ref{cor:OptimalEtaDer}, the first derivative of $\eta(\gamma)$ is negative, and hence an optimal $\gamma$ has to meet the bound in~\eqref{eq:constgamma}, i.e.,  $\sum^N_{i=1}\gamma_i = \sigma$. We thus consider $\gamma$ as a vector parameter  in $\Sp[\sigma]$, the simplex of height~$\sigma$ in $\R^N$ defined in~\eqref{eq:defsimplex}.  
With the preliminaries above, we establish our first main result, captured by the following theorem:

\begin{Theorem}\label{thm:MTHM1} 
Let  $\eta : \Sp[\sigma] \longrightarrow \R$  be defined in~\eqref{eq:cost1}. Then, under Assumption~\ref{asmp:fornetworksystem}, the following hold:
\begin{enumerate}
\item For $\epsilon$ small, $\eta$ is strictly convex over $\Sp[\sigma]$. 
\item A point $\gamma^* \in \Sp[\sigma]$ is the unique minimum of $\eta$  if and only if there exists a number $\mu \le 0$ such that
$$
\displaystyle \left. \frac{\partial \eta(\gamma)}{\partial \gamma_i} \right |_{\gamma = \gamma^*}= \mu \hspace{10pt} \mbox{if }  \gamma_i \neq 0, 
$$
and
$$
\displaystyle \left. \frac{\partial \eta(\gamma)}{\partial \gamma_j} \right |_{\gamma = \gamma^*} \ge \mu \hspace{10pt} \mbox{if }  \gamma_j = 0.  
 $$
\end{enumerate}\,
\end{Theorem}

\begin{proof} 
We first prove item~1. Note that the cost function $\eta(\gamma)$ implicitly depends on $\epsilon$ via $A$ in the ARE~\eqref{eq:RiccatiEquationNetSys}. We thus write $\eta(\gamma; \epsilon)$ explicitly to emphasize the dependence.  Now, suppose that $\epsilon = 0$; then, the $N$ linear sub-systems are  decoupled. In this case, the cost function $\eta(\gamma; 0)$ is also decoupled and  can be written as the sum of independent functions of the entries $\gamma_i$ as $$
\eta(\gamma; 0) = \sum^N_{i=1} \eta_{i}(\gamma_i) 
$$
where $\eta_{i}(\gamma_i)$ is given by 
$$
\eta_{i}(\gamma_i) =  \tr\( \Sigma_i \),
$$  
for $\Sigma_i$ the positive definite solution to the  ARE
\begin{equation*}
A_i \Sigma_i + \Sigma_i A_i^\top - (\gamma_i + r_{\min}) \Sigma_i c_i c_i^\top \Sigma_i + G_iG_i^\top = 0.  
\end{equation*}
Since each triplet $(A_i, c_i, G_iG_i^\top )$ is regular,  we know from Corollary~\ref{cor:OptimalEtaDer}  that $\eta''_{i}(\gamma_i) > 0$ for all $\gamma_i > 0$, and hence
$$
\eta''(\gamma; 0) = 
\diag\(
\eta''_{1}(\gamma_1), \ldots, \eta''_{N}(\gamma_N) \) > 0, \hspace{10pt} \forall\, \gamma \in \Sp[\sigma], 
$$
which implies that $\eta(\gamma; 0)$ is strictly convex.  
Note that the matrix $A$ (defined in~\eqref{eq:defANetSys}) is linear in $\epsilon$, and the matrix $\Sigma$ in~\eqref{eq:RiccatiEquationNetSys} are analytic in $A$~\cite{delchamps_analytic1984}. Hence, the cost function $\eta(\gamma; \epsilon)$ is analytic in $\epsilon$. In particular, if $\epsilon$ is sufficiently small, then $\eta(\gamma; \epsilon)$ is strictly convex. 
The second item of Theorem~1 then directly follows the Karush-Kuhn-Tucker (KKT) conditions~\cite{luenberger1997optimization}. 
This completes the proof. 
\end{proof}

\section{Monotonicity properties and filtration of optimal resource allocation}\label{sec:monotone}

While Theorem~\ref{thm:MTHM1} is only valid in the $\tau \to 0$ asymptotic, in practical situations, we can expect that the asymptotic regime is reached for relatively low sampling frequencies (which are related to the largest real part of the eigenvalues of $A$) per the approximation of~\eqref{eq:covupdate}. In this section, we assume that the equivalence of system~\eqref{eq:ModelWeaklyCoupled1} and~\eqref{eq:ModelWeaklyCoupled} holds, and we investigate how the unique minimum of $\eta(\gamma)$ varies  as the total available number of slots within a time-period $\rt$ (and thus $\sigma$) increases.
 
 To this end, we introduce the set\begin{equation}\label{eq:defconvexhull}
\Ch[\sigma_{\max}] := \sqcup_{0\le \sigma \le \sigma_{\max}}\Sp[\sigma]. 
\end{equation}
It is  clear that $\Ch[\sigma_{\max}]$ is the convex hull of the  origin in $\R^N$ and  the points $\sigma_{\max} e_1,\ldots, \sigma_{\max} e_N$, where $e_1,\ldots, e_N$ form the canonical basis of $\R^N$. Now, we define the function $ f: [0,\sigma_{\max}] \longrightarrow \Ch[\sigma_{\max}]$ as follows:
 \begin{equation}\label{eq:littlef} f(\sigma) := \arg\min_{\gamma \in \Sp[\sigma]} \eta(\gamma).\end{equation}  We call $f$ {\bf the optimal allocation map}.  By construction, $f(\sigma) \in \Sp[\sigma]$.

We show in this section, starting from the model given in~\eqref{eq:MMODEL}, that the optimal allocation $f(\sigma)$ is  ``well-behaved'' with respect to the increase in the total amount of resource $\sigma$ in the following sense: the sampling rate allocated to each sensor by the optimal strategy is {\it nondecreasing} in the total sampling rate. Moreover, it is strictly increasing if and only if the total sampling rate exceeds a certain threshold.  
To make this notion precise, we first borrow the following definition from measure theory: 

\begin{Definition}
Let $S$ be an arbitrary  set, and $\cal{P}(S)$ be the {\bf power set} of $S$. A map $\mathcal{F}: \R \longrightarrow \cal{P}(S)$ is  a {\bf filtration} if for $x_1 \leq x_2$, we have that $\cal{F}(x_1) \subseteq \cal{F}(x_2) $. 
\end{Definition}
 
The map $f$ is said to be {\bf monotonically increasing} if it is \emph{entry-wise} monotonically increasing, i.e., if we let $f_i$, for $i = 1,\ldots, N$, be the $i$-th entry of $f$, then, for $\sigma_1\le \sigma_2$,  
$$
f_i(\sigma_1) \le f_i(\sigma_2), \hspace{10pt} \forall \, i = 1,\ldots, N.
$$
We show in this section that the map $f$ is monotonically increasing provided that the subsystems~\eqref{eq:ModelWeaklyCoupled}  are weakly coupled. 

For ease of notation, let $\cal{I}:=\{1,\ldots, N\}$ be the index set. For a $\sigma \ge 0$, let $\cal{I}_\sigma$ be a subset of $\cal{I}$ defined as follows:  
$$
\cal{I}_\sigma := \{ i \in \cal{I} \mid f_i(\sigma) > 0 \}. 
$$
By definition,  if $j \notin \cal{I}_\sigma$, then $f_j(\sigma) = 0$. In other words, 
 $\cal{I}_{\sigma}$ is comprised of the indices of nonzero entries of $f(\sigma)$. 
We then  consider the map 
\begin{equation}\label{eq:bigF}
\mathcal{F}: [0, \sigma_{\max}] \longrightarrow \cal{P}(\cal{I}),
\end{equation} 
which sends $\sigma$, the total available rate, to $\cal{I}_{\sigma}$. Note that if we can show that $f$ is monotonically increasing,  then it follows immediately that $\cal{F}$ is a filtration. 

We now formalize the results stated above. First, recall that from Theorem~1, for a fixed $\sigma\in [0,\sigma_{\max}]$, the derivative $\eta'(\gamma)$ at  $\gamma^*(\sigma)=f(\sigma)$  satisfies the following condition: there exists a number $\mu_{\sigma}$ such that 
$$
\left\{
\begin{array}{ll}
\left. \frac{\partial \eta(\gamma)}{\partial \gamma_i} \right |_{\gamma = f(\sigma)} = \mu_\sigma, & \forall\, i\in \cal{I}_{\sigma},  \vspace{3pt}\\
 \left. \frac{\partial \eta(\gamma)}{\partial \gamma_j} \right |_{\gamma = f(\sigma)} \ge \mu_\sigma, & \forall \, j\notin \cal{I}_{\sigma}. 
\end{array}
\right.
 $$
We now collect together the indices for which we have an equality in the relation above: define the subset $\cal{J}_{\sigma}\subset \cal{I}$ to be 
\begin{equation}\label{eq:defJsigma}
\cal{J}_{\sigma}:= \left\{ i\in \cal{I} \mid \left. \frac{\partial \eta(\gamma)}{\partial \gamma_i }\right|_{\gamma = f(\sigma) }  = \mu_\sigma  \right\}.
\end{equation}
By construction,  we have $$\cal{I}_{\sigma} \subseteq \cal{J}_{\sigma},$$ 
and moreover, the equality $\cal{I}_{\sigma} = \cal{J}_{\sigma} $
 holds if and only if 
 $$
\left. \frac{\partial \eta(\gamma)}{\partial \gamma_j}  \right |_{\gamma = f(\sigma)} >  \mu_\sigma, \hspace{10pt} \forall \, j\notin \cal{I}_{\sigma}. 
$$
We will see soon that the equality indeed holds for almost all $\sigma \in [0,\sigma_{\max}]$. On the other hand, there are $\sigma$'s for which the equality does {\em not} hold. Specifically, these are the $\sigma$'s for which there exists a $j\in\cal{I}$ such that $$\left. \frac{\partial \eta(\gamma)}{\partial \gamma_j}\right |_{\gamma = f(\sigma)} =  \mu_\sigma \hspace{5pt} \mbox{ and } \hspace{5pt} f_j(\sigma)=0.$$ 
We thus collect these $\sigma$'s and define the set $\cal D$ as follows: 
\begin{equation}\label{eq:Alisuggested}
\cal{D} := \{ \sigma \in [0,\sigma_{\max}] \mid \cal{I}_{\sigma} \subsetneq \cal{J}_\sigma\} \cup \{0, \sigma_{\max}\}. 
\end{equation}
We now show that $\cal{D}$ is a {\em finite} set, and moreover, each $\sigma\in \S - \{0, \sigma_{\max}\}$ is a point of ``discontinuity" of the map~$\cal{F}$.  
Precisely, we establish below our second main result, captured by the following theorem: 

\begin{Theorem}\label{thm:MTHM2}
Let the maps $f$ and $\cal{F}$ be defined in~\eqref{eq:littlef} and~\eqref{eq:bigF}, respectively, and let the set $\cal{D}$ be defined in~\eqref{eq:Alisuggested}. Then, 
for $\epsilon$ sufficiently small, the following hold: 
\begin{enumerate}
\item There are only finitely many points in $\cal{D}$, which we label as 
$$
0 = \sigma_0 < \sigma_1 < \ldots < \sigma_m = \sigma_{\max}.  
$$
\item The map $f$ is continuous and monotonically increasing. Moreover, $f$ is continuously differentiable over each open interval $(\sigma_{i-1}, \sigma_{i})$ for $i = 1,\ldots, m$. 
\item The map $\mathcal{F}$ is a filtration. In particular, we have
$$
\varnothing = \cal{I}_{\sigma_0} \subsetneq \cal{I}_{\sigma_1} \subsetneq \ldots \subsetneq \cal{I}_{\sigma_m},
$$ 
and for each $i = 1,\ldots, m$, we have
$$
\cal{I}_\sigma = \cal{I}_{\sigma_i}, \hspace{10pt} \forall\, \sigma \in (\sigma_{i-1},  \sigma_i]. 
$$ 
\end{enumerate}\,
\end{Theorem}

In the remainder of this section, we establish the properties of the maps $f$ and $\cal{F}$ that are needed to prove Theorem~\ref{thm:MTHM2}.

\subsection{On right-differentiability of~$f$}\label{sec:rightdifferentiable}
 We first recall the definition of right-differentiability:
\begin{Definition}\label{def:rightdifferentiability}
Let $f:[a, b] \longrightarrow \R^N$ be an arbitrary  function defined over a closed interval $[a, b]$ of $\R$. We say that $f$ is {\bf right-continuous} at $x\in [a,b)$ if 
$$
\lim_{\varepsilon \to 0, \varepsilon > 0} f(x+\varepsilon) = f(x),
$$ 
and is {\bf right-differentiable} at $x$ if the limit
$$
\partial_+ f(x):=\lim_{\varepsilon \to 0, \varepsilon > 0} \frac{f(x+\varepsilon) - f(x)}{\varepsilon}
$$ 
exists. We call $\partial_+ f(x)$ the {\bf right-derivative} of $f$ at $x$. 
\end{Definition} 


Now, let $f$ be the optimal allocation map  allocation. We compute in this subsection the right-derivative of $f$. To proceed, we first 
state a fact about the Hessian matrix $$\eta''(\gamma; \epsilon):= \frac{\partial^2 \eta(\gamma;\epsilon)}{\partial \gamma^2}.$$ . 

\begin{lem}\label{lem:convexfunctionoverslices}
Let $\sigma_{\max}>0$ and $\Ch[\sigma_{\max}]$ be the convex set defined in~\eqref{eq:defconvexhull}.  
If the coupling strength $|\epsilon|$ is sufficiently small, then for any  $\gamma\in\Ch[\sigma_{\max}]$ and any principal submatrix $M$ of  $\eta''(\gamma,\epsilon)$, $M^{-1}$ exists. It is moreover positive definite and diagonally dominant.
\end{lem}

\begin{proof}
We have shown in the proof of Theorem~\ref{thm:MTHM1} that for $\epsilon=0$ and for all $\gamma \in \Ch[\sigma_{\max}]$
$$
\eta''(\gamma;0) = 
\diag\(
\eta''_{1}(\gamma_1), \ldots, \eta''_{N}(\gamma_N) \) > 0. 
$$ 
Using a classical argument involving the continuity of eigenvalues with respect to the entries of a matrix and the compactness of $\Ch[\sigma_{\max}]$,  we conclude that for $\epsilon>0$ sufficiently small, Lemma~\ref{lem:convexfunctionoverslices} holds.
\end{proof}

We will implicitly assume, for the remainder of this section, that $|\epsilon|$ is  small enough so that Lemma~\ref{lem:convexfunctionoverslices} holds. With such a choice of $\epsilon$, we  derive the following fact about the right-derivative of $f$:
 
\begin{pro}\label{pro:derivativeofwgammastar}
The right-derivative $\partial_+f(\sigma)$ exists for all $\sigma$ in $[0,\sigma_{\max})$, and is right-continuous. Moreover,     
\begin{equation*}
\left\{
\begin{array}{ll}
\partial_+ f_i(\sigma)  > 0 &   \mbox{if } \hspace{5pt} i\in \cal{J}_\sigma,\\
\partial_+ f_i(\sigma)  = 0 &   \mbox{otherwise}.
\end{array}\right. 
\end{equation*}
where we recall that the definition of ${\cal J}_\sigma$ is given in~\eqref{eq:defJsigma}.
\end{pro}

As a consequence of Proposition~\ref{pro:derivativeofwgammastar}, we have
\begin{cor}\label{cor:atpalazzo}
For each $\sigma \in [0, \sigma_{\max})$, there exists a number $\delta > 0$ such that for all $\sigma' \in (\sigma, \sigma + \delta)$, 
\begin{equation}\label{eq:IequalJ}
\cal{I}_{\sigma'} = \cal{J}_{\sigma'} = \cal{J}_{\sigma}. 
\end{equation}
Moreover, for any such $\delta>0$, the map $f$ is continuously differentiable and monotonically increasing over $[\sigma, \sigma + \delta)$.  
\end{cor}

We omit the proof of Corollary~\ref{cor:atpalazzo} as it directly follows from the definition of right-continuity and Proposition~\ref{pro:derivativeofwgammastar}. 

In the remainder of the subsection, we prove Proposition~\ref{pro:derivativeofwgammastar}.  Fix $\sigma$ in $[0, \sigma_{\max})$ and, without loss of any generality,  assume that $$\cal{J}_{\sigma} = \{1,\ldots, k\}.$$

We now introduce helpful notation for the analysis below. Let $x^*_\sigma \in \R^k$ be the first~$k$ entries of $f(\sigma)$, i.e., we have  $$f(\sigma)  = (x^*_\sigma, \0).$$ 

Note that $x^*_\sigma$ may also have zero entries since, from the definitions of $\cal{I}_\sigma$ and $\cal{J}_\sigma$,  the entries of $x^*_\sigma$ are all nonzero if and only if $\cal{I}_\sigma = \cal{J}_\sigma$.  For a vector $x\in \R^k$, we set $\ol x \in \R^N$ to be the vector obtained  by appending $(N - k)$ zero entries at the end of the vector:  
$$
\ol x := (x,\0) \in \R^N.  
$$
In particular, we note here that from the definition, $\ol {x^*_\sigma} = f(\sigma)$. 
Conversely, for a vector $y = (y_1,\ldots, y_N)\in \R^N$, we truncate $y$ by defining 
$$y_{[1:k]}:= (y_1,\ldots, y_k) \in \R^k.$$ 

Now let the open neighborhood $U$ of $x^*_\sigma$ be chosen such that for any $x = (x_1,\ldots, x_k)\in U$, we have $x_i > -r_{\min}$ for all $i = 1,\ldots, k$. Then, from Lemma~\ref{lem:domainofeta}, $\eta(\ol x) = \tr(\Sigma(\ol x))$ is well defined for all $x\in U$.  We set $$\phi(x):= {\eta'(\ol x)}_{[1:k]}.
$$ 

 We show below that if $U$ is sufficiently small, then $\phi$ is a diffeomorphism.  
To establish this fact, we show that the derivative of the map $\phi$, is full rank. To this end, we partition the Hessian matrix $\eta''(\ol x)$ into $2\times 2$ blocks as follows:
$$
\eta''(\overline{x}) = 
\begin{pmatrix}
\eta''_{11}(\overline{x}) & \eta''_{ 12}(\overline{x})\\
\eta''_{ 21}(\overline{x}) & \eta''_{22}(\overline{x})
\end{pmatrix}
$$
where $\eta''_{11}(\overline{x})$ is a $k\times k$ matrix. Then, by definition of $\phi$, we obtain 
$$
\phi'(x) = \eta''_{11}(\overline{x}), 
$$  
and in particular, 
$
\phi'(x^*_\sigma) = \eta''_{11}(f(\sigma))  
$. 
From Lemma~\ref{lem:convexfunctionoverslices}, we obtain the following result: 

\begin{lem}\label{lem:determineopensetU}
If the open neighborhood $U$ of $x^*_\sigma$ is sufficiently small, then for any $x\in U$, the inverse $\phi'(x)^{-1}$ exists, and is positive definite and diagonally dominant. In particular, $\phi$ is a diffeomorphism.
\end{lem}

\begin{proof}
From Lemma~\ref{lem:convexfunctionoverslices},  $\eta''_{11}(\ol x^*_\sigma)^{-1}$ is positive definite and diagonally dominant. Lemma~\ref{lem:determineopensetU} then follows from the fact that $\phi''(x)$ depends smoothly on~$x\in U$. The second statement follows from the inverse function theorem.
\end{proof}

For the remainder of the subsection, we assume that the open neighborhood $U$ of $x^*_\sigma$ is chosen such that Lemma~\ref{lem:determineopensetU} holds. 

From item~2 of Theorem~\ref{thm:MTHM1} and the definition of $\cal{J}_\sigma$ (see~\eqref{eq:defJsigma}), we have
$$
\phi(x^*_\sigma) = \eta'(f(\sigma))_{[1:k]} = \mu_\sigma {\bf 1}, 
$$ 
where $\bf 1$ is a vector of all ones in $\R^k$. Denote by $V\subset \R^k$  the image of $U$ by $\phi$. Then, $V$ is an open set, and hence there exists an $\varepsilon > 0$ such that the following closed-open line segment: 
$$
l_{\varepsilon} := \{ \mu {\bf 1} \mid  \mu \in [\mu_\sigma, \mu_\sigma + \varepsilon) \} 
$$
is contained in $V$. 

Now, by using the fact that $\phi$ is a diffeomorphism between $U$ to $V$, we know that the map  
$\left. \phi^{-1} \right |_{l_\varepsilon}$,  
is a diffeomorphism between $l_{\varepsilon}$ and its image, which is a one-dimensional curve in $U$.  
We denote this curve as follows:
\begin{equation}\label{eq:definexmu}
x(\mu) := \left. \phi^{-1} \right |_{l_\varepsilon}(\mu) \in U,
\end{equation}
and let $x_i(\mu)$ be the $i$-th entry of $x(\mu)$. Further, we define a function $s: [\mu_\sigma, \mu_\sigma + \varepsilon) \longrightarrow \R$ as the sum of the entries of $x(\mu)$:   
\begin{equation}\label{eq:definesmu}
s(\mu) := \sum^k_{i = 1} x_i(\mu). 
\end{equation}
Note that by their definitions, both $x(\mu)$ and $s(\mu)$ are smooth; we denote their derivatives by $$x'(\mu) := \frac{d x(\mu)}{d\mu} \hspace{10pt} and \hspace{10pt} s'(\mu) := \frac{d s(\mu)}{ d\mu}.$$  
We further write by 
$
x'(\mu )  \succ 0
$ 
if $x'_i(\mu) > 0$ for all $i = 1,\ldots, k$. With the preliminaries above, 
we  establish relationships among $x(\mu)$, $s(\mu)$ and the  optimal  allocation map:

\begin{lem}\label{lem:parametrizedbymu}
The following hold for $x(\mu)$ and $s(\mu)$: 
\begin{enumerate}
\item For each $\mu \in [\mu_\sigma, \mu_\sigma + \varepsilon)$,  
$$x'(\mu)\succ 0 \hspace{10pt} \mbox{and} \hspace{10pt} s'(\mu) > 0.$$
\item If $\varepsilon$ is sufficiently small, then 
$$
\ol x(\mu) =  f(s(\mu)).  
$$
\end{enumerate}\,
\end{lem}  

\begin{proof}
We first establish item~1 of Lemma~\ref{lem:parametrizedbymu}.  By definition of $x(\mu)$, we have
$$
x'(\mu) = \frac{d }{d \mu} \left. \phi^{-1}  \right |_{l_\varepsilon}(\mu), 
$$
which can be evaluated as follows:
$$
x'(\mu) = \left [\phi'(x(\mu))\right ]^{-1} \mathbf{1} \succ 0.
$$
The inequality in the equation above holds because from Lemma~\ref{lem:determineopensetU}, $\phi'(x)^{-1}$ is positive definite and \emph{diagonally dominant} for all $x\in U$. 
It then follows that $$s'(\mu) = \sum^k_{i = 1} x'_i(\mu)> 0.$$ 

We now establish the second item. First, note that since $x'(\mu) \succ 0$, we have $$\ol x(\mu) \succ \ol x(\mu_\sigma)\succeq \0\in \R^N,$$ for all $\mu > \mu_\sigma$.    
 Furthermore, by the fact that 
$$
\sum^N_{i=1} \ol x_i(\mu) = \sum^k_{i= 1} x_i(\mu) = s(\mu),
$$
we obtain that $$\ol x(\mu) \in \Sp[s(\mu)].$$ Next, note that from the definition of the map $\phi^{-1} |_{l_{\varepsilon}}$, we have 
$$
\phi(x(\mu)) = \phi\(\left. \phi^{-1} \right |_{l_{\varepsilon}}(\mu)\) = \mu \mathbf{1},   
$$
and since $\phi(x(\mu)) = \eta'(\ol x(\mu))$, we obtain 
\begin{equation*}\label{eq:condition1}
\left. \frac{\partial \eta(\gamma)}{\partial \gamma_i}  \right |_{\gamma = \ol x(\mu)} = \mu, \hspace{10pt} \forall \, i = 1,\ldots,k.  
\end{equation*}
So, from item~2 of Theorem~1, it suffices to show that if  $\varepsilon$ is sufficiently small, then
\begin{equation}\label{eq:condition2}
\left. \frac{\partial \eta(\gamma)}{\partial \gamma_j} \right |_{\gamma = \ol x(\mu)} \ge \mu, \hspace{10pt} \forall\, j = k+1, \ldots, N,   
\end{equation}
for any $\mu \in [\mu_\sigma, \mu_\sigma + \varepsilon)$. But this holds because we have assumed that  
$\cal{J}_{\sigma} =\{1,\ldots, k\}$, and by  definition of $\cal{J}_\sigma$, the inequalities in~\eqref{eq:condition2} are strict when $\mu = \mu_\sigma$. Now, since $ \eta'(\ol x(\mu)) = \phi(x(\mu))$ is smooth in $\mu$, we conclude that there is an $\varepsilon > 0$ such that~\eqref{eq:condition2} holds for all $\mu \in [\mu_\sigma, \mu_\sigma + \varepsilon)$. This completes the proof.    
\end{proof}


With Lemma~\ref{lem:parametrizedbymu},  we prove Proposition~\ref{pro:derivativeofwgammastar}:

\begin{proof}[Proof of Proposition~\ref{pro:derivativeofwgammastar}] 
Let $\sigma \in [0,\sigma_{\max})$ and $\varepsilon>0$ be small enough so that Lemma~\ref{lem:parametrizedbymu} holds.  We compute the right-derivative of $f$ at $\sigma$.  
The proof relies on the use of the two smooth maps $x(\mu)$ and $s(\mu)$ (defined in~\eqref{eq:definexmu} and~\eqref{eq:definesmu}, respectively), both of which are defined over the  interval $[\mu_\sigma, \mu_\sigma + \varepsilon)$.  

First, note that by definition, 
$s(\mu_\sigma) = \sigma$, and 
from item~1 of Lemma~\ref{lem:parametrizedbymu}, $s(\mu)$ is \emph{strictly} monotonically increasing in $\mu$. Hence,   there is a $\delta > 0$ such that the map  
$$
s: [\mu_\sigma,  \mu_\sigma + \varepsilon) \longrightarrow [\sigma, \sigma + \delta), 
$$
with $s(\mu_\sigma) = \sigma$,  is a diffeomorphism. Moreover, its inverse
$
s^{-1}: [\sigma, \sigma + \delta ) \longrightarrow [\mu_\sigma, \mu_\sigma + \varepsilon), 
$
is also smooth and strictly monotonically increasing. Now, from item~2 of Lemma~\ref{lem:parametrizedbymu}, we have that for $\sigma' \in [\sigma,\sigma + \delta)$, 
$$
f(\sigma')  = \ol x(s^{-1}(\sigma')), 
$$
and hence $f(\sigma')$ is smooth in $\sigma' \in [\sigma, \sigma + \delta)$. Moreover, by the chain rule, 
$$
\partial_+ f_i(\sigma) = 
\left\{ 
\begin{array}{ll}
x'_i(\mu_\sigma) / s'(\mu_\sigma) > 0  &  \mbox{if } i = 1,\ldots,k, \\ 
0 & \mbox{if } i = k+1, \ldots, N,
\end{array}
\right.
$$
where the  inequality above follows directly from the first item of Lemma~\ref{lem:parametrizedbymu}. 
\end{proof}

\subsection{Analysis and Proof of Theorem~\ref{thm:MTHM2}}\label{sec:leftdifferentiable}

In this subsection, we establish Theorem~\ref{thm:MTHM2}.  Recall that from Corollary~\ref{cor:atpalazzo}, for each $\sigma\in [0, \sigma_{\max})$, there exists a number $\delta$ such that 
\begin{equation}\label{eq:definingcondition}
\cal{I}_{\sigma'} = \cal{J}_{\sigma'} = \cal{J}_{\sigma}, \hspace{10pt} \forall\, \sigma' \in (\sigma, \sigma + \delta).
\end{equation}
Now, let $\delta_\sigma$ be the supremum among all such $\delta$:  
\begin{equation*}
\delta_\sigma :=  \sup \{\delta \in \R \mid \mbox{Eq.}~\eqref{eq:definingcondition} \mbox{ holds} \}.  
\end{equation*}
We introduce the map $\rho$ :  
\begin{equation}\label{eq:deftau}
\rho: \sigma \mapsto \min\{ \sigma + \delta_\sigma, \sigma_{\max} \}. 
\end{equation}
Appealing again to Corollary~\ref{cor:atpalazzo}, we know that the map $f$ 
is continuously differentiable and monotonically increasing over $[\sigma, \rho(\sigma))$.  We now show that the map $f$ is continuous at $\rho(\sigma)$. First, note that from Proposition~\ref{pro:derivativeofwgammastar}, $f$ is  continuous from the right, it thus suffices to show that $f$ is continuous from the left. Precisely, we establish the following result:

\begin{pro}\label{pro:propertiesoftau}
Let $\sigma\in [0, \sigma_{\max})$, and $\rho(\sigma) \in (\sigma, \sigma_{\max}]$ be defined in~\eqref{eq:deftau}. Then, 
\begin{equation}\label{eq:proofofcontinuity}
\lim_{\varepsilon \to 0, \varepsilon > 0} f(\rho(\sigma) - \varepsilon )  = f(\rho(\sigma)). 
\end{equation}\,
\end{pro} 

\begin{Remark}
By combining~Corollary~\ref{cor:atpalazzo} and Proposition~\ref{pro:propertiesoftau}, we have
$$
\cal{I}_{\sigma'}= \cal{I}_{\rho(\sigma)} = \cal{J}_{\sigma} \subseteq \cal{J}_{\rho(\sigma)}, 
$$
for all $\sigma' \in (\sigma, \rho(\sigma)]$. Furthermore, by the definition of the map~$\rho$, we have that if $\rho(\sigma) \neq \sigma_{\max}$, then the last inequality in the equation above is strict. 
\end{Remark}

\begin{proof}[Proof of Proposition~\ref{pro:propertiesoftau}]
We first show that the limit in~\eqref{eq:proofofcontinuity} exists.  
Without loss of generality, we assume that $\cal{J}_\sigma = \{1,\ldots, k\}$, for $k\le n$. Then, by the definition of $\rho(\sigma)$ and from Proposition~\ref{pro:derivativeofwgammastar}, we have that $f(\sigma')$ is smooth over $[\sigma, \rho(\sigma))$, and moreover,  for all $\sigma' \in [\sigma,\rho(\sigma))$, 
$$
\left\{
\begin{array}{ll}
\partial_+  f_i(\sigma') > 0 & \mbox{if }  i= 1,\ldots, k, \\
\partial_+  f_i(\sigma') = 0  & \mbox{if } i = k+1, \ldots, N. 
\end{array}   
\right. 
$$ 
In particular, each $f_i(\sigma')$, for $i = 1,\ldots, k$, is strictly monotonically increasing over $[\sigma,\rho(\sigma))$. Then, combining the fact that each $f_i(\sigma')$, for $i = 1,\ldots, N$, is nonnegative and the fact that  for all $\sigma' \in[\sigma, \rho(\sigma)) $, 
$$
f_i(\sigma') \le \sum^N_{i = 1} f_i(\sigma') = \sigma' < \rho(\sigma), 
$$ 
we have that 
$
\lim_{\varepsilon \to 0, \varepsilon > 0} f(\rho(\sigma) - \varepsilon )  
$ 
exists. 

Now, let the limit be denoted by $\widetilde \gamma$.   It should be clear that the nonzero entries of $\widetilde \gamma$ are the first $k$ entries. 
We show below that $f(\rho(\sigma)) = \widetilde \gamma $. First, note that by the definition of the map $\rho$ (see~\eqref{eq:deftau}), we have
$$\cal{I}_{\sigma'} = \cal{J}_{\sigma'}= \{1,\ldots, k\}$$
for any $\sigma'\in (\sigma, \rho(\sigma))$,  
and hence
$$
\left\{
\begin{array}{ll}
\left. \frac{\partial \eta(\gamma)}{\partial \gamma_i} \right |_{\gamma = f(\sigma')}  = \mu_{\sigma'} & \forall\, i = 1, \ldots, k, \vspace{3pt}\\
\left. \frac{\partial \eta(\gamma)}{\partial \gamma_i} \right |_{\gamma = f(\sigma')} > \mu_{\sigma'} & \forall \, i = k+1, \ldots, N. 
\end{array}
\right.
 $$
Then, using the fact that $\eta'(\gamma)$ is smooth in $\gamma$ and the fact that $\widetilde \gamma = \lim_{\varepsilon \to 0, \varepsilon > 0} f(\rho(\sigma) - \varepsilon )$, we know that there exists a $\widetilde \mu \in \R$ such that   
$$
\left\{
\begin{array}{ll}
\left. \frac{\partial \eta(\gamma)}{\partial \gamma_i} \right |_{\gamma = \widetilde \gamma}  = \widetilde\mu & \forall\, i = 1, \ldots, k, \vspace{3pt}\\
\left. \frac{\partial \eta(\gamma)}{\partial \gamma_i} \right |_{\gamma = \widetilde \gamma} \ge \widetilde \mu & \forall \, i = k+1, \ldots, N. 
\end{array}
\right.
 $$
We thus conclude, from item~2 of Theorem~1, that $f(\rho(\sigma)) = \widetilde \gamma$. It then follows that  for all $\sigma' \in (\sigma, \rho(\sigma)]$, 
$$
\cal{J}_{\sigma} = \cal{I}_{\sigma'} = \cal{I}_{\rho(\sigma)} = \{1,\ldots, k\},
$$
and $\cal{J}_\sigma \subseteq \cal{J}_{\rho(\sigma)}$.  
\end{proof}

We are now in a position to prove Theorem~\ref{thm:MTHM2}. 

\begin{proof}
The proof of Theorem~\ref{thm:MTHM2} relies on the use of the map~$\rho$ and the properties established in Proposition~\ref{pro:propertiesoftau}.  
We first show that the set $\cal{D}$, defined in~\eqref{eq:Alisuggested}, is finite. Let $\sigma_0 := 0$; then, for a positive integer $l>0$, we define a number $\sigma_l$ in $[0,\sigma_{\max}]$ as follows:
$$
\sigma_l := \underbrace{\rho\circ \rho \circ \cdots \circ \rho}_{l\mbox{ times}}(\sigma_0),  
$$
i.e., we iteratively apply  $\rho$ for $l$ times. 

We  now show that there exists an integer $m > 0$ such that 
\begin{equation}\label{eq:eq1inproof}
0 = \sigma_0 < \sigma_1 < \ldots < \sigma_m = \sigma_{\max},
\end{equation} 
and hence $\sigma_{m+1}$ does not exist. 
Note that if~\eqref{eq:eq1inproof} holds, then from the definition of the map~$\rho$, we have 
$$
\cal{D} = \{\sigma_0,\ldots,\sigma_m\}. 
$$ 
To establish~\eqref{eq:eq1inproof}, first note that $\sigma_1 = \rho(\sigma_0)$ is well defined. Since $\sigma_1 > 0$, we obtain
$$
\varnothing = \cal{I}_{\sigma_0} \subsetneq \cal{I}_{\sigma_1}. 
$$
If $\sigma_1 = \sigma_{\max}$, then~\eqref{eq:eq1inproof} holds.  We thus assume that $\sigma_1 < \sigma_{\max}$. But then,  $\sigma_2 = \rho(\sigma_1)$ is well defined, 
and moreover, from Proposition~\ref{pro:propertiesoftau}, we have
$$
\cal{I}_{\sigma_{1}} \subsetneq \cal{J}_{\sigma_1}  = \cal{I}_{\sigma_2}. 
$$ 
Combining the two inequalities above, we obtain 
$$
\varnothing = \cal{I}_{\sigma_0} \subsetneq \cal{I}_{\sigma_1} \subsetneq  \cal{I}_{\sigma_2}.
$$
So, by repeatedly applying  the arguments above, we obtain a sequence of inequalities as follows:  
$$
\varnothing = \cal{I}_{\sigma_0} \subsetneq \cal{I}_{\sigma_1} \subsetneq \ldots \subsetneq \cal{I}_{\sigma_l} \subsetneq \ldots ;  
$$
but since $\cal{I}$ is a finite set, the chain {\em has to} terminate in a finite number of steps. In other words, there must exist a positive integer $m$ such that $\sigma_m = \rho(\sigma_{m-1}) = \sigma_{\max}$.  
We have thus proved the first part of Theorem~\ref{thm:MTHM2}. 

The last two parts of Theorem~\ref{thm:MTHM2} then directly follow from Corollary~\ref{cor:atpalazzo} and Proposition~\ref{pro:propertiesoftau}. Indeed, from Corollary~\ref{cor:atpalazzo}, the map $f$ is smooth and strictly monotonically increasing over $[\sigma_{i-1},\sigma_i)$ for all $i = 1,\ldots, m$. Then, from Proposition~\ref{pro:propertiesoftau}, for all $i = 1,\ldots, m$, 
$$
\lim_{\varepsilon \to 0, \varepsilon > 0} f(\sigma_i - \varepsilon) = f(\sigma_i),  
$$
and hence $f$ is continuous over the entire interval $[0,\sigma_{\max}]$. Furthermore, appealing again to Proposition~\ref{pro:propertiesoftau}, 
we have that for all $i= 1,\ldots, m$, 
$$
\cal{I}_{\sigma_{i-1}} \subsetneq \cal{I}_{\sigma'}  = \cal{I}_{\sigma_i}, \hspace{10pt} \forall\, \sigma' \in (\sigma_{i-1}, \sigma_i]. 
$$
This completes the proof. 
\end{proof}

\section{Conclusions}

We considered in this paper an optimal resource allocation problem for the MSE estimation of a networked system. Precisely,  a network of $N$ weakly, dynamically coupled linear systems is connected via a shared communication channel  to a network manager. Each system can send sampled measurements of its own state over this channel. These samples are then used by the network manager  to estimate the global state of the network. Given that the channel capacity is finite, the problem is how to optimally schedule the transmission of the samples. Such problems arise in the control and estimation of real-time cyber-physical systems~\cite{lureal,seto2001trade, saifullah2014near}

In the set-up employed, we set aside information-theoretic considerations and made the simplifying assumption that each sample requires the same amount of communication, and  the capacity is thus directly proportional to the number of samples sent per unit of time. The analysis went along the following steps: for a unit of time $\tau$, we assume that we have a number $\rt$ of time slots, and that one sample can be sent per time slot. The optimal allocation problem is then a problem of designing an optimal strategy for  assigning the $\rt$ time slots to $N$ different subsystems. We have then shown in Section~\ref{sec:prelim} that  this optimal allocation problem is equivalent, in the limit $\tau \to 0$, to a problem of optimal allocating sensor qualities---namely, the signal-to-noise ratio of the measurements they provide---albeit for continuous-time sample measurements. 
 In particular, we have shown that in the limit $\tau \to 0$, the order in which the $\rt$ slots are assigned to the $N$ subsystems over a time period $\tau$ becomes irrelevant, but only the number of slots assigned  to each subsystem matters. We have then used the above-mentioned equivalence and shown in Section~III that the optimal allocation problem, in the limit $\tau\to 0$, is a strictly convex optimization problem under certain mild assumptions. 
We have further studied in Section~IV how the optimal allocation strategy evolves as the total channel capacity increases.



\bibliographystyle{IEEEtran}
\bibliography{OptimalAllocation}

\section*{Appendix A}
We prove here Proposition~\ref{pro:existenceoflimit}. The proof will be carried out by constructing and solving an optimal control problem. 
Let $R$ be an allocation strategy. We define $R'\in \cal{R}$ be reversing the order of $R$, i.e., $$R'_{\rt - l}:= R_{l+1}, \hspace{10pt} \forall l = 0,\ldots, \rt - 1.$$   
Let $\cal{A}_0:= e^{A^\top \tau_0}$.   We consider the following discrete-time periodically-switched control system:  
\begin{equation}\label{eq:controlsys}
x_{[l + 1]}[k] = \cal{A}_0 x_{[l]}[k] + \bar c_{R'_{l}} u_{[l]}[k], \hspace{10pt} l = 0,\ldots,\rt - 1.
\end{equation}
So, each time step $k$ is comprised of $\rt$ sub-steps (we identify $x_{[\rt ]}[k]$ with $x_{[0]}[k+1]$). There are $N$ controllers $\bar c_{1},\ldots, \bar c_N$ in total. But only one controller can be used at a single sub-step, which is determined by the reversed strategy $R'$. 

We now introduce a cost function associated with the control system~\eqref{eq:controlsys}: Let $H$ be a positive definite matrix given by $$H := \int^{\tau_0}_0 e^{As} GG^\top e^{A^\top s } ds.$$ 
For each $i = 1,\ldots, N$, we define a matrix $P_i$ as follows:
$$
P_i := 
\begin{bmatrix}
H & H \bar c_i \\
\bar c_i^\top H & \bar c_i^\top H \bar c_i + I 
\end{bmatrix}
$$   
Note that each matrix $P_i$ is positive definite; indeed, for a vector $z = (x, u)$ with $x$ and $u$ in appropriate dimensions, we have 
$$
z^\top P_i z = \|\sqrt{H} (x + \bar c_i u) \| + \|u\|,
$$ 
where $\|\cdot\|$ is the standard Euclidean norm,  and $\sqrt{H} > 0$ is the square root of $H$. Thus, $z^\top P_i z = 0$ if and only if $z = 0$.  
With the matrices $P_i$'s defined above, we define a finite horizon cost function as follows:
$$
\xi_T :=  x^{\top}_{[1]}[T] K_0 x_{[1]}[T]  +  \sum^{T-1}_{k = 0} \sum^{\rt}_{l = 1} z^\top_{[l]}[k] P_{R'_l} z^\top_{[l]}[k],
$$ 
where $K_0 > 0$, and $z_{[l]}[k] := (x_{[l]}[k], u_{[l]}[k])$.

This is a classical optimal control problem, which can be solved by introducing a Hamiltonian of the system and writing down the co-state equations. We omit the details here, but provide the solution to the problem. To this end,  we define a sequence of positive definite matrices $\{K_{[l]}[k] \mid 0 \le k \le T - 1, 0\le l \le \rt - 1\}$ by backward recursion: For the initialization, we let $K_{[0]}[T] := K_0$. Then, for the recursion, we have   
\begin{equation}\label{eq:updateforK}
K_{[l]}[k] = \[ \[\cal{A}^\top_0 K_{[l+1]}[k] \cal{A}_0  + H  \]^{-1} + \bar  c_{R'_{l}} \bar c_{R'_{l}}^\top \]^{-1},
\end{equation}  
where $K_{[\rt ]}[k]$ is  identified with $K_{[0]}[k+1]$. We note here that~\eqref{eq:updateforK} can be re-written into a discrete-time dynamic Riccati equation via the use of the Woodbury matrix identity. 
Now, with the matrices $K_{[l]}[k]$'s defined above, we have that the optimal controls $u^*_{[l]}[k]$'s are given by:
\begin{equation*}
u^*_{[l]}[k] := -\( I + \bar c^\top_{R'_{l}} H \bar c^\top_{R'_{l}} + \bar c^\top_{R'_{l}} K_{[l+1]}[k] \bar c_{R'_{l}} \)^{-1} \\
\(\bar c^\top_{R'_{l}} H  + \bar c^\top_{R'_{l}} K_{[l+1]}[k] \cal{A}_0\) x_{[l]}[k],
\end{equation*} 
and the corresponding cost function is simply given by
$$
\xi^*_T := x^\top_{[0]}[0] K_{[0]}[0] x_{[0]}[0]. 
$$

Recall that the map $\phi_{i}$ is defined by combining~\eqref{eq:covupdate} and~\eqref{eq:covplusupdate}, which sends an error covariance matrix to its update over a single slot. By comparing~\eqref{eq:updateforK} with~\eqref{eq:covupdate} and~\eqref{eq:covplusupdate}, we obtain that 
$
K_{[l-1]}[k] = \phi_{R'_l} K_{[l]}[k]
$. 
Thus, if we let $K[k]:= K_{[0]}[k]$, then 
$
K[k-1] = \rho_{R'_1}\cdots\rho_{R'_{\rt}}K[k]
$.     
Since $R'$ is defined by reversing the order of $R$, we have
$$
K[k - 1] =  \rho_{R_{\rt}}\cdots\rho_{R_{1}}(K[k]) =  \Phi_R(K[k]). 
$$
It then follows that 
$
K[T-k] = \Phi^{k}_R(K_0)  
$. 
In particular, if we let $K_0$ be the initial condition of error covariance $\Sigma(0)$, then $K[T-k] = \Sigma(k)$. 
On the other hand, it is well known that if system~\eqref{eq:controlsys} is controllable (a sufficient condition of controllability will be established shortly), then the optimal control problem introduced above can be solved for an infinite horizon (i.e., we let  $T$ go to infinity); indeed, if the system is controllable, then one is able to drive the system back to the origin in finite time steps. It then implies that the minimal cost $\xi^*_\infty$ exists, and so does the limit $\lim_{k\to\infty} \Phi^k_R(K_0)$. Furthermore, we note that the limit does not depend on the initial condition $K_0$ because it is necessary that the optimal controls $u^*_{[l]}[k]$ drive the states $x_{[l]}[k]$ converge to $0$, and hence the term $\lim_{T\to\infty}x^\top_{[0]}[T]K_0 x_{[0]}[T]$ in the cost function is zero regardless of the value of $K_0$.     

It thus suffices for us to show that if the assumption of Proposition~\ref{pro:existenceoflimit} is satisfied, i.e., $(A, \bar c)$ is observable and $\tau \in \cal{T}$, then system~\eqref{eq:controlsys} is controllable. First, we define matrices $\cal{A}$ and $\cal{B}$ as follows:
$$
\cal{A}:= e^{A^\top\tau} = \cal{A}_0^{\rt}
$$
and
$$
\cal{B}:= \[\bar c_{R'_1}, \cal{A}_0 \bar c_{R'_2}, \ldots, \cal{A}_0^{\rt -1} \bar c_{R'_{\rt}}\].
$$
Then, it follows that system~\eqref{eq:controlsys} is controllable if and only if $(\cal{A}, \cal{B})$ is controllable. We let $C(\cal{A}, \cal{B})$ be the controllability matrix, and ${\rm Col}(\cal{A}, \cal{B})$ be the column space of $C(\cal{A}, \cal{B})$. 
Recall that $r_i$ is the number of slots assigned to a sensor~$i$ over a period $\tau$, and $r_i \ge r_{\min} > 0$.  Thus, for each~$i = 1,\ldots, N$, there exists an $l_i$ such that $R'_{l_i} = i$. It should be clear that 
\begin{equation*}
{\rm Col}(\cal{A}, \cal{B}) \supseteq  \sum^{N}_{i = 1}{\rm Col}\(\cal{A},  \cal{A}_0^{l_i-1} \bar c_{R'_{l_i}}\)  \\ =  \sum^{N}_{i = 1}\cal{A}_0^{l_i-1} {\rm Col}\(\cal{A},  \bar c_{i}\) 
\end{equation*}   
Furthermore, by appealing to the Cayley-Hamilton theorem, we have that 
\begin{equation}\label{eq:ssss}
\cal{A}_0^{l_i-1} {\rm Col}\(\cal{A},  \bar c_{i}\)  \subseteq {\rm Col}\(A^\top,\ol c_i\)
\end{equation} 
Note that if $\lambda$ is an eigenvalue of $A$, then $e^{\lambda\tau}$ is an eigenvalue of $\cal{A}$. Furthermore, if $\tau\in \cal{T}$, then $e^{\lambda_i \tau} \neq e^{\lambda_j\tau}$ for any two distinct eigenvalues $\lambda_i$ and $\lambda_j$ of~$A$, and hence  
the (generalized) eigenspace of $A$ corresponding to the eigenvalue $\lambda$ is the same as the (generalized) eigenspace of $\cal{A}$ corresponding to the eigenvalue $e^{\lambda\tau}$. Thus, a sufficient condition for the equality in~\eqref{eq:ssss} to hold is that $
\tau \in \cal{T}$. As a consequence, it follows that 
\begin{equation*}
{\rm Col}(\cal{A}, \cal{B})  \supseteq 
\sum^{N}_{i = 1}\cal{A}_0^{l_i-1} {\rm Col}\(\cal{A},  \bar c_{i}\)  \\ = \sum^{N}_{i = 1} {\rm Col}\(A^\top,\ol c_i\) = {\rm Col}\(A^\top, \ol c\). 
\end{equation*}
Since $(A, \ol c)$ is observable,  we have ${\rm Col}\(A^\top, \ol c\) = \R^n$, and hence
${\rm Col}(\cal{A}, \cal{B}) = \R^n$, 
which implies that $(\cal{A}, \cal{B}) $ is controllable.

\section*{Appendix B}


\setcounter{subsection}{0}

We prove here Proposition~\ref{pro:oneforall}, and establish the fact that a triplet $(A, c, P)$, with $(A, c)$ observable and $P > 0$, is generically regular.  
Recall that  $\Sigma(r)$, for $r > 0$, is the unique positive definite solution to the following ARE:
\begin{equation}\label{eq:riccatieqinappendix}
A \Sigma(r) + \Sigma(r) A^\top - r \Sigma(r) \, cc^\top \Sigma(r) + P = 0, 
\end{equation}
and $\Sigma'(r)$ is the derivative of $\Sigma(r)$ with respect to $r$. 
We first establish the following fact:

\begin{pro}\label{pro:equivalentstatements} 
Let $(A, c)$ be an observable pair, and $P > 0$. Then, the following three items are equivalent:
\begin{enumerate}
\item The triplet $(A, c, P)$ is regular.
\item For all $r > 0$, $\Sigma'(r) < 0$.
\item For all $r > 0$, $(A, \Sigma(r) c)$ is controllable. 
\end{enumerate}\,
\end{pro}

\begin{proof}

First, note that from Definition~\ref{def:regulartriplet}, item~3 implies item~1. We show below that first, items~2 and~3 are equivalent,  and second, item~1 implies item~2.

To show that items~2 and~3 are equivalent, it suffices to show that for any fixed $r > 0$, $\Sigma'(r) < 0$ if and only if $(A, \Sigma(r)c)$ is controllable. Recall that $\Sigma'(r)$ satisfies the following Lyapunov equation:
\begin{equation*}
\(A - r \Sigma cc^\top \) \Sigma' + \Sigma' \(A^\top  - r cc^\top\Sigma \) - \Sigma cc^\top \Sigma = 0,
\end{equation*}
and hence has the following explicit expression:  
\begin{equation*}\label{eq:K'(wgamma)lyapunov}
\Sigma' = -\int^{\infty}_{0} e^{\(A - r \Sigma cc^\top  \) t}\, \Sigma cc^\top \Sigma e^{\(A^\top - r cc^\top \Sigma \) t} dt. 
\end{equation*}
This, in particular,  implies that $\Sigma'< 0$ if and only if $\(A - r \Sigma cc^\top, \Sigma c \)$ is controllable, which holds if and only if $\(A, \Sigma c \)$ is controllable.

We now show that item~1 implies item~2. Since $(A, c, P)$ is a regular triplet, from Definition~\ref{def:regulartriplet}, there exists a $\widetilde r > 0$ such that $(A, \Sigma (\widetilde r)c)$ is controllable, and hence,  from the arguments above, $\Sigma'(\widetilde r) < 0$. We need to show that $\Sigma'(r) < 0$ for all $r > 0$. 
The proof will be carried out by contradiction: we show that if there is a number $r_1> 0$ such that $\Sigma'(r_1)$ is singular, with $\Sigma'(r_1) v = 0$ for some nonzero vector $v\in \R^n$, then 
$\Sigma'(r) v = 0$ for all $r \ge 0$. To see this, recall that from Proposition~\ref{lem:SigmaDerivative}, 
$$\Sigma'(r) \le 0 \hspace{5pt} \mbox{ and } \hspace{5pt} \Sigma''(r) \ge 0, \hspace{10pt} \forall\, r > 0.$$ 
Then, since $v^\top \Sigma'(r_1) v = 0$, we have that for any $r \ge r_1$, 
$$
0 \ge v^\top \Sigma'(r) v = \int^r_{r_1} v^\top \Sigma''(s)v\,  ds \ge 0.  
$$
It then follows that $v^\top \Sigma'(r) v = 0$, which in turn implies that $\Sigma'(r)v = 0$ because $\Sigma'(r) \le  0$.  
On the other hand, $\Sigma(r)$ is analytic in $r$. So, if $\Sigma'(r) v= 0 $ for all $r \ge r'$, then $\Sigma'(r)v = 0$ for all $r > 0$. But this contradicts the fact that $\Sigma'(\widetilde r) < 0$. We have thus shown that item~1 implies item~2, which completes the proof.  
\end{proof}

Proposition~\ref{pro:oneforall} then immediately follows from Proposition~\ref{pro:equivalentstatements}. We are now in a position to establish the genericity of regular triplets. 


\begin{pro}\label{pro:regulartripareopendense}
Let $\cal{X}$ be the set of triplets $\(A, c, P\)$, with $(A, c)$ observable  and $P > 0$. Let $\cal{X}_{\rm reg} \subset \cal{X}$ be the set of regular triplets. 
Then, $\cal{X}_{\rm reg}$ is open and dense in $\cal{X}$. 
 \end{pro}
 
 \begin{proof}
 First, we show that $\cal{X}_{\rm reg}$ is open in $\cal{X}$. Let $r$ be a  positive number, and let $\Sigma(r)$ be the unique positive definite solution to the ARE~\eqref{eq:riccatieqinappendix}, where $(A, c, P)$ is a regular triplet. Then, from Proposition~\ref{pro:oneforall}, we have that $\Sigma'(r) < 0$. Now, fix the number $r$, and we perturb the triplet $(A,c,P)$ in~\eqref{eq:riccatieqinappendix}: since $\Sigma$ is analytic in $(A, c, P)$,      there is  an open neighborhood $U$ of $(A, c, P)$ in $\cal{X}$ such that the inequality $\Sigma'(r) < 0$ still holds even if we replace $(A,c,P)$ in~\eqref{eq:riccatieqinappendix} with an arbitrary triplet $(A',c',P')$ in $U$. Appealing again to Proposition~\ref{pro:oneforall}, we have that each $(A', c', P')$ in $U$ is a regular triplet.  This then shows that $\cal{X}_{\rm reg}$ is open in $\cal{X}$.


 We now show that  $\cal{X}_{\rm reg}$ is dense in $\cal{X}$. To do so, we construct a regular triplet arbitrarily close to an arbitrary triplet.  First, perturb matrix $A$, if necessary,  so that  $A$ is diagonalizable and the eigenvalues of $A$ are not repeated. Since observability is an open condition, we can choose a perturbation small enough so that  it remains observable after the perturbation.  
We also note that for any such matrix $A$,  the pair $\(A, v\)$ is controllable for almost all $v$ in $\R^n$. This latter fact  implies  that there exists a symmetric matrix $\delta \Sigma$, with $\|\delta \Sigma\|$ arbitrarily small, such that $\(\Sigma + \delta \Sigma\)$ is positive definite and 
 $\(A, (\Sigma + \delta \Sigma) c \)$ is controllable. Now, fix any such $\delta \Sigma$, and let $\delta P$ be a symmetric matrix defined as follows: 
 $$
 \delta P:= -\( A - \Sigma cc^\top  \)  \delta \Sigma - \delta \Sigma \( A^\top - r cc^\top \Sigma \), 
 $$  
 Note that $\|\delta P\|$ can be made arbitrarily small by decreasing $\|\delta \Sigma\|$, and hence we can assume that the matrix $\(P + \delta P\) $ is positive definite. 
 Now, let 
 $$
 \widetilde \Sigma := \Sigma + \delta \Sigma \hspace{10pt} \mbox{ and } \hspace{10pt} \widetilde P := P + \delta P.
 $$
Then, from construction, $\widetilde \Sigma$ is the unique positive definite solution to the following ARE:
$$
A \widetilde \Sigma +  \widetilde \Sigma A^\top - r  \widetilde \Sigma cc^\top  \widetilde \Sigma + \widetilde P= 0. 
$$  
Since $(A,  \widetilde \Sigma c)$ is controllable, from Proposition~\ref{pro:oneforall}, $(A, c,  \widetilde P )$ is a regular triplet. This then shows that $\cal{X}_{\rm reg}$ is dense in $\cal{X}$. 
\end{proof}

\section*{Appendix C}  

We prove here Proposition~\ref{lem:SigmaDerivative}:

\begin{proof}[Proof of Proposition~\ref{lem:SigmaDerivative}]
To simplify the notation, we suppress the explicit dependence of $\Sigma$ on $r$. First, we show that $\Sigma' \le 0$ and $\Sigma''\ge 0$. Let
$$
\overline{A} := - A \hspace{5pt} \mbox{ and } \hspace{5pt}  
K := \Sigma^{-1}.  
$$
Then, $K$ satisfies the following ARE:
$$
\ol{A}^\top K + K \ol{A} - K GG^\top K + r cc^\top = 0.  
$$ 
Differentiating $K$ with respect to $r$, we obtain $K'$ as
\begin{equation}\label{eq:saysomethingaboutP'andP''really1}
(\ol{A}^\top - KGG^\top ) K'  + K' \(\ol{A} - GG^\top K \) + cc^\top = 0,
\end{equation}
and $K''$ as
\begin{equation}\label{eq:saysomethingaboutP'andP''really2}
(\ol{A}^\top - KGG^\top )   K''  + K'' \(\ol{A} - GG^\top K \) \\- 2 K' GG^\top K' = 0. 
\end{equation}

Let $\sqrt{GG^\top}$ be the square root of $GG^\top$, which is the unique positive definite matrix such that $\sqrt{GG^\top}\sqrt{GG^\top}=GG^\top$. Note that $\sqrt{GG^\top}$ is of full rank, and hence $\(\ol{A}, \sqrt{GG^\top}\)$ is controllable, which implies that 
$\(\ol{A}^\top - KGG^\top \) $ is stable.  So, from~\eqref{eq:saysomethingaboutP'andP''really1} and~\eqref{eq:saysomethingaboutP'andP''really2}, we have that
$$
K' \ge 0 \hspace{5pt} \mbox{ and } \hspace{5pt} K'' \le 0. 
$$
Furthermore, by differentiating the equality $K\Sigma = I$ with respect to~$r$, we get $K\Sigma'=-K'\Sigma$, which implies 
\begin{equation}\label{eq:derivativesofS}
\left\{
\begin{array}{l}
\Sigma'  = -\Sigma K' \Sigma \leq 0  \\
\Sigma'' =  2 \Sigma K' \Sigma K' \Sigma -  \Sigma K'' \Sigma \geq 0.
\end{array}
\right. 
\end{equation} 
We used the facts that $K'\le 0$ for the first inequality, and $\Sigma > 0$ and $K'' \leq 0$ for the second inequality.

We now assume that the triplet $\(A, c, GG^\top\)$ is regular, and show that $\Sigma' < 0$ and $\Sigma'' > 0$. By differentiating~\eqref{eq:LyapunovEqSingle1}, we obtain the Lyapunov equation:
$$
\(A- r \Sigma cc^\top\) \Sigma' + \Sigma'\(A^\top - r cc^\top \Sigma\)- \Sigma c c^\top \Sigma = 0, 
$$
which admits the solution:
$$
\Sigma' = -\int^{\infty}_0 e^{\(A - r \Sigma cc^\top\) t}\, \Sigma c c^\top \Sigma \, e^{\(A^\top - r  cc^\top \Sigma\) t}\, dt.
$$
Since the triplet $(A, c, GG^\top)$ is regular, we know, from Proposition~\ref{pro:oneforall}, that the pair $\(A,   \Sigma c\)$ is controllable. This in turn implies that  $\(A  - r \Sigma cc^\top, \Sigma c \)$ is controllable. Hence, we have that  $\Sigma' < 0$. Finally, using the equalities in ~\eqref{eq:derivativesofS}, we conclude that 
$K' = -K \Sigma' K > 0$, 
and hence 
$$\Sigma'' \ge 2\Sigma K' \Sigma K' \Sigma > 0,$$
which completes the proof. 
\end{proof}

\section*{Appendix D}\label{sec:geometricmethod}
We  set-up here a gradient algorithm to locate the local minima $\gamma^*$ of $\eta$ over $\Sp[\sigma]$. The algorithm can be shown to be globally convergent and can be used instead of convex optimization methods.
   
Denote by $M$ the set of symmetric, rank one matrices with unit trace:
$$
M : = \{ vv^\top \mid v\in \R^n,  \|v\| = 1\}.  
$$ 
Because the diagonal entries of $vv^\top$ are $v_i^2$, we have that $\sum_{i=1}^n v_i^2 =1,$  and thus the projection  map $\pi: M \longrightarrow \Sp[\sigma]$ given by 
$$H \mapsto \sigma \(h_{11},\cdots,h_{nn} \)^\top$$ 
is well defined, where we let $H=(h_{ij})$. It is easily seen to be surjective as well. Hence, there exists a function $\Phi: M \longrightarrow \R$ such that 
$$\Phi(H):= \eta(\pi(H)).$$ 
To be more explicit, we first recall that $\eta(\gamma)$ is defined as $
\eta(\gamma) = \tr(\Sigma)
$, where $\Sigma$ is the unique positive definite solution to the ARE: 
\begin{equation}\label{eq:defKss}
A\Sigma + \Sigma A^\top -\Sigma c_{\gamma} c^\top_\gamma \Sigma  + GG^\top = 0, 
\end{equation}
with $c_{\gamma}$  given by
$$
c_\gamma = \diag\(\sqrt{\gamma_{1}+r_{\min}} c_1,\ldots, \sqrt{\gamma_{N}+r_{\min}} c_N \).
$$
So, if we let $\gamma = \pi(H)$, for all $i = 1,\ldots, n$, then $\Phi(H)$ depends on $H$ via $\Sigma$. 
Our goal in this section is to derive a gradient algorithm for $\Phi$. To this end, we first need to define a metric on $M$. This is done in the following paragraph.
\paragraph{Normal metric and  double bracket flows} Let $\so(n)$ be the set of $n$-by-$n$ skew-symmetric matrices: $\so(n) = \{ \Omega \in \R^{n \times n} \mid \Omega = -\Omega^\top\}$. 
For a point $H$  in $M$, let $T_HM$ be the tangent space of $M$ at $H$.  Then, it is a well-known fact that
$$
T_H M = \left \{[H, \Omega] \mid \Omega \in \so(n) \right \},  
$$ 
where $[A, B] := AB - BA$ is the commutator of matrices. We also adopt the standard notation $\ad_H \Omega := [H,\Omega]$. Note that if the symmetric matrix $H$ has pairwise distinct eigenvalues, then $\ad_H$ is invertible~\cite{brockett_diffgeomgradesign93}. Then, the so-called normal metric is defined as follows: for any two elements $X$ and $Y$ in $T_H M$, let
\begin{equation}\label{eq:normalmetric}
g_H\(X, Y\) := -\tr \(\ad^{-1}_H X\, \ad^{-1}_H Y \).
\end{equation}

In our setup, however, $H$ will be of rank one and thus have multiple zero eigenvalues, whence $\ad_H$ has a nontrivial kernel. We thus need to slightly adapt the definition of the normal metric to handle this case.  The modification goes as follows (see also~\cite{BeGeom2016}):  
Let $\ker_H$ be the kernel of $\ad_H$, and $\ker^{\perp}_H$ the subspace of $\so(n)$ orthogonal, with respect to the Frobenius norm, to $\ker_H$, i.e., $\|\Omega\|_F = \sqrt{\tr\(\Omega^\top\Omega\)}$. Then, 
$
ad_H$ is a linear isomorphism when restricted to $\ker^{\perp}_H$. With a slight abuse of notation, we introduce the map 
\begin{equation}\label{eq:adHonkernal}
\ad_H : \ker^{\perp}_H \longrightarrow T_HM,  
\end{equation}
and denote its inverse as $\ad^{-1}_H$.  The normal metric~\eqref{eq:normalmetric} is then well defined on $M$ provided that we use the definition of $\ad_H$ in Eqn.~\eqref{eq:adHonkernal}. Let $\Phi$ be a smooth function over $M$, and denote by 
$$
\Phi'(H) : = \( \frac{\partial \Phi(H) }{ \partial h_{ij}}\)_{ij}\in \R^{n\times n}. 
$$ 
We then have the following result, adapted from~\cite{brockett_diffgeomgradesign93} :

\begin{lem}\label{lem:doublebracketflow} 
 The gradient flow of $\Phi(H)$ on $M$ with respect to the normal metric~\eqref{eq:normalmetric} is given by  
\begin{equation}\label{eq:doublebracketflow}
\frac{d}{dt}H = -\left [H,\left [H, \Phi'(H)\right ] \right ].
\end{equation}
Moreover, $H$ is an equilibrium of the gradient flow if and only if $[H,\Phi'(H)] = 0$. 
\end{lem}
\vspace{5pt}

\paragraph{Double bracket flow for optimal resource allocation} We  build upon the results of the previous paragraphs to introduce a  differential equation whose solutions provably converge to an  optimal allocation vector~$\gamma^* \in \Sp[\sigma]$. Namely, we will derive the gradient flow of the function $\Phi(H)$ over the space $M$---from Lemma~\ref{lem:doublebracketflow}, it suffices for us to compute $\Phi'(H)$. 
From the definition of the projection map~$\pi$, the potential function $\Phi(H)$ depends {\it only on} the diagonal entries of $H$. 
As a consequence, $\Phi'(H)$ is a diagonal matrix; 
indeed, because $\Phi(H) = \eta(\pi(H))$, we have
\begin{equation}\label{eq:derivativeofPhi}
 \Phi'(H) = \sigma \, \diag\( \eta'( \pi(H))\)
 \end{equation}
 where we recall that $\eta'(\pi(H))$ is the derivative $\partial \eta(\gamma)/\partial \gamma $ evaluated at $\pi(H)$. 
 We thus appeal to Lemma~\ref{lem:doublebracketflow} and derive the double bracket gradient-descent of $\Phi(H)$ as follows:
\begin{equation}\label{eq:gradientflowforAllocation}
\frac{d}{dt}H = -\sigma \left [H,\left [H, \diag\( \eta'(\pi(H))  \)  \right ] \right ].
\end{equation}
We elaborate here the evaluation of $\eta'(\pi(H))$. For ease of notation, denote by $\gamma := \pi(H)$.  
 Let $\Sigma'_i: =\partial \Sigma/ \partial \gamma_{i}$, with $\Sigma$ the positive definite solution to the ARE~\eqref{eq:defKss}. Then,
$$
\frac{\partial \eta(\gamma)}{\partial \gamma_{i} } =  \tr\( \Sigma'_i \). 
$$ 
It thus remains to compute $\Sigma'_i$. To this end, let $C_i$ be an $nN\times nN$ matrix defined as follows: First, divide $C_i$ into $N\times N$ blocks, with each block an $n\times n$ matrix. Then, let  the $ii$-th block of $C_i$ be $c_ic_i^\top$, and all the other blocks be zero matrices.  Now, by a simple calculation using~\eqref{eq:defKss}, we obtain $\Sigma'_i$ as the solution to the following Lyapunov equation:  
$$
\(A - \Sigma c_{\gamma}c_{\gamma}^\top  \) \Sigma'_i+   \Sigma'_i \(A^\top -  c_{\gamma}c_{\gamma}^\top \Sigma \) - \Sigma C_i \Sigma = 0.
$$

The convergence of the double bracket flow follows from the fact that the cost function $\eta(\gamma)$ defined over $\Sp[\sigma]$, and hence $\Phi(H)$ defined over $M$, has a unique local minimum point.
\end{document}